\documentclass[lettersize,journal]{IEEEtran}
\usepackage{cite}
\usepackage{amsmath,amssymb,amsfonts}
\makeatletter
\def\maketag@@@#1{\hbox{\m@th\normalfont\normalsize#1}}
\makeatother
\usepackage{graphicx}
\usepackage{textcomp}

\usepackage{graphics}
\usepackage{epstopdf}
\usepackage{epsfig}
\usepackage{mathrsfs}
\usepackage{color}
\usepackage{array}
\usepackage{booktabs}
\usepackage{lipsum,multicol}
\usepackage{subfigure}
\usepackage{times}
\usepackage{setspace}

\usepackage{algorithm}
\usepackage{pdfsync}
\usepackage{pgfplots}
\usepackage{pgfplotstable}
\usepackage[subnum]{cases}
\usepackage{multirow}
\usepackage{algpseudocode}
\usepackage{float}
\usepackage{setspace}
\usepackage{bm}
\usepackage{multicol}
\usepackage{amsthm}
\usepackage{epstopdf}
\usepackage[center]{caption}
\newtheorem{theorem}{Theorem}
\newtheorem{lemma}{Lemma}

\newtheorem{remark}{Remark}

\newcommand{\beq}{\begin{equation}}
	\newcommand{\eeq}{\end{equation}}

\newcommand{\paren}[1]{\left(#1\right)}

\newcommand{\field}[1]{\ensuremath{\mathbb{#1}}}
\newcommand{\R}{\ensuremath{\field{R}}} 

\def\BibTeX{{\rm B\kern-.05em{\sc i\kern-.025em b}\kern-.08em
		T\kern-.1667em\lower.7ex\hbox{E}\kern-.125emX}}

\algdef{SE}[REPEATUNTIL]{RepeatUntil}{EndRepeat}[1]{\qquad \; \algorithmicrepeat\ #1}{\qquad \; \algorithmicuntil\ \text{Convergence}}
\begin{document}
	\title{Optimal Privacy-Aware Stochastic Sampling}
	\author{Chuanghong Weng, Ehsan Nekouei \vspace{-2em} \thanks{
			C. Weng and E. Nekouei are with the Department of Electrical Engineering, City University of Hong Kong (e-mail: cweng7-c@my.cityu.edu.hk; enekouei@cityu.edu.hk). 
			
			The work was partially supported by the Research Grants Council of Hong Kong under Project CityU 21208921, a grant from Chow Sang Sang Group Research Fund sponsored by Chow Sang Sang Holdings International Limited.
	}}
	\maketitle
	\begin{abstract}
		This paper presents a stochastic sampling framework for privacy-aware data sharing, where a sensor observes a process correlated with private information. A sampler determines whether to retain or discard sensor observations, balancing the tradeoff between data utility and privacy. Retained samples are shared with an adversary who may attempt to infer the private process, with privacy leakage quantified using mutual information. The sampler design is formulated as an optimization problem with two objectives: $\left(\romannumeral1\right)$ minimizing the reconstruction error of the observed process using the sampler's output, $\left(\romannumeral2\right)$ reducing the privacy leakages. For a general class of processes, we show that the optimal reconstruction policy is deterministic and derive the optimality conditions for the sampling policy using a dynamic decomposition method, which enables the sampler to  control the adversary's belief about private inputs. For linear Gaussian processes, we propose a simplified design by restricting the sampling policy to a specific collection, providing analytical expressions for the reconstruction error, belief state, and sampling objectives based on conditional means and covariances. Additionally, we develop a numerical optimization algorithm to optimize the sampling and reconstruction policies, wherein the policy gradient theorem for the optimal sampling design is derived based on the implicit function theorem. Simulations demonstrate the effectiveness of the proposed method in achieving accurate state reconstruction, privacy protection, and data size reduction.
	\end{abstract}
	
	
	\section{Introduction} \label{sec:introduction}
	\subsection{Motivation}
	The increasing proliferation of the Internet of Things (IoT) has revolutionized data collection, processing, and utilization across various domains, including healthcare, smart cities, and industrial automation. However, this pervasive data collection introduces not only data storage challenges but also significant privacy concerns, as the information gathered by IoT device sensors often contains sensitive and private details about individuals or environments. 
	
	Moreover, IoT devices are constrained by limited resources, including battery life and computational capacity, necessitating efficient communication strategies to conserve energy and reduce data transmission. To address these constraints, IoT devices are expected to transmit system information only when necessary, thereby mitigating privacy risks, reducing energy consumption, and minimizing storage requirements.
	
	Despite these requirements, most existing sampling and transmission strategies for IoT privacy protection are tailored to scenarios where the system state is private but needs to be shared with a remote third party. However, these methods often fail to address critical challenges in cases where the third party may be malicious and capable of inferring sensitive information. Additionally, system variables may be divided into two categories—public and private—that evolve through different processes \cite{li2018information, jia2017privacy}, further limiting the applicability of these strategies. Another significant challenge arises from adversaries who might continuously refine their estimates of past private states using newly acquired observations over time. As a result, effective privacy measures must consider not only the leakage of the current state but also the reduced uncertainty about past private states. However, most existing approaches measure privacy leakage solely based on the error covariance of current state estimator, overlooking the substantial risks posed by the disclosure of past private states.
	
	To address these limitations, we propose a privacy-aware sampling mechanism based on an information-theoretic privacy metric. This mechanism is designed to achieve two key objectives:  $\left(\romannumeral1\right)$ to prevent private information from being reliably inferred from the sampled data, and $\left(\romannumeral2\right)$ to preserve the utility of the sampled data while using fewer samples.
	\subsection{Contributions}
	\begin{figure}
		\centering
		\includegraphics[width=0.45\textwidth]{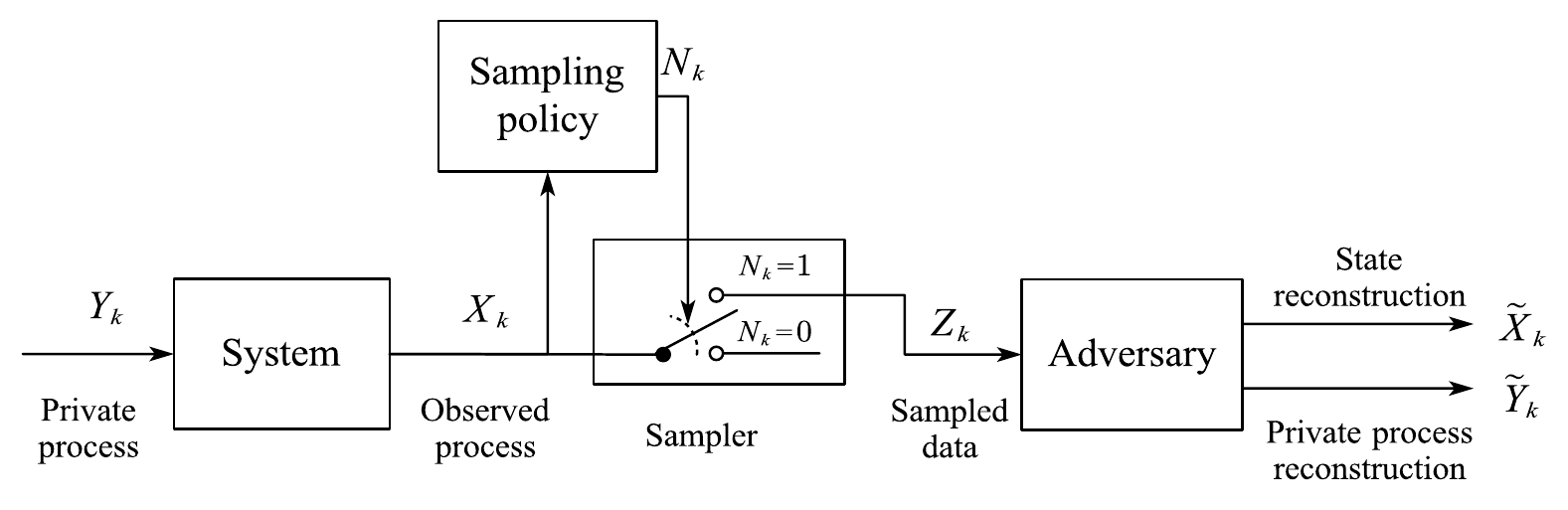}
		\caption{The privacy-aware stochastic sampling setup.}
		\label{Fig.EstInfSys}
	\end{figure}
	This paper addresses the problem of privacy-aware stochastic sampling in the context of the system setup depicted in Fig. \ref{Fig.EstInfSys}. Specifically, a sensor observes a process $X_k$ that is influenced by an underlying private process $Y_k$. Observations of $X_k$ may be shared with an untrusted third party (referred to as the adversary), which poses a risk of the private process $Y_k$ being inferred. To mitigate this risk, we explore the optimal design of a stochastic sampling mechanism that selectively determines whether to retain or discard samples of $X_k$. Moreover, we propose a reconstruction policy to recover $X_k$ from the retained samples.
	
	To quantify privacy, we employ mutual information between the private process $Y_k$ and the sampler's output as the privacy metric. We frame the co-design of the privacy-aware sampling and reconstruction policies as an optimization problem, where the objective is to minimize a weighted combination of the mutual information (privacy leakage) and the reconstruction error of $X_k$.
	
	The primary contributions of this paper are summarized as follows:
		\begin{enumerate}
			\item Optimal Sampling Policy via Dynamic Programming Decomposition:
			We derive the backward optimality equation for the sampling policy using a dynamic programming decomposition approach (Theorem \ref{Th.OPTEQU}). This enables the optimal sampling policy to efficiently regulate the adversary's belief about the private process.
			\item Simplified Sampling Policy for Linear Gaussian Processes:
			For systems governed by linear Gaussian processes, we present a simplified sampling design. This method facilitates recursive computation of the optimal reconstruction policy using the conditional mean and covariance matrix. Additionally, we derive a closed-form expression for the privacy metric (mutual information), significantly reducing computational complexity. Furthermore, we demonstrate that the optimal sampling policy controls the adversary's conditional covariance of the private process (Theorem \ref{Th.OPTEQULIN}).
			\item Policy Gradient Algorithm:
			To jointly optimize the sampling and reconstruction policies, we develop a policy gradient algorithm based on the implicit function theorem (Algorithm \ref{Alg.ETESTOPT}).
			\item Simulation Experiments:
			We validate the proposed co-design framework through simulation experiments. The results confirm the framework's effectiveness in enhancing privacy protection while simultaneously reducing the data size.
	\end{enumerate}
	\subsection{Related Work}
	Privacy protection in dynamic systems has been addressed using perturbed additive noise techniques, such as the differential privacy framework proposed in \cite{dwork2006calibrating}. In \cite{le2013differentially}, differential privacy was applied to design privacy-aware Kalman filters and approximate privacy-aware stable filters. Bayesian differential privacy for linear systems was introduced in \cite{sugiura2021bayesian}, where sufficient conditions for privacy via additive noise were derived. The interplay between input observability and differential privacy was explored in \cite{kawano2020design}, which also proposed a privacy-preserving controller design. Additionally, a differentially private sampling approach was proposed in \cite{le2019differentially} to protect the state of a linear system.
	
	Another line of research focuses on information-theoretic metrics for privacy in dynamic systems. In \cite{tanaka2017directed} and \cite{nekouei2019information}, the co-design of controllers and privacy-preserving filters was studied for LQG systems to minimize quadratic costs and privacy leakage measured by directed information. These works showed that the optimal private filter combines Kalman filtering with additive Gaussian noise, while the control law remains linear. The authors in \cite{molloy2023smoother} addressed state trajectory obfuscation in partially observable Markov processes using the conditional entropy of states given measurements and control signals. Mutual information has also been employed as a privacy measure: the authors in \cite{li2018information} proposed privacy-preserving smart meter mechanisms using rechargeable batteries, while the authors in \cite{erdemir2020privacy} introduced actor-critic algorithms to optimize privacy-aware data release policies. The relationship between differential privacy and information-theoretic approaches was analyzed in \cite{unsal2023information}.
	
	While previous work on privacy-aware state estimation and closed-loop control \cite{tanaka2017directed, molloy2023smoother, li2018information, erdemir2020privacy} have largely focused on noise injection and state obfuscation, this paper explores a stochastic sampling approach to information-theoretic privacy. Our method not only enhances privacy but also reduces data transmission and storage requirements. Furthermore, we compare our approach with additive noise methods and demonstrate competitive utility-privacy trade-offs, while achieving greater efficiency in data transmission.
	
	We note that deterministic transmission and sampling strategies for secure communication have also been widely studied, especially for linear systems. In \cite{leong2018transmission}, a deterministic transmission policy was designed to balance estimation accuracy at a remote estimator and privacy leakage in presence of an adversary, and numerical algorithms were provided to obtain the transmission policy. A related work \cite{wang2022transmission} optimized deterministic transmission policies for encrypted and plaintext data, leveraging monotonic scheduling sequences for linear systems. The authors in \cite{leong2019information} used directed information to measure communication security and formulated optimization problems to enhance security while reducing estimation errors.
	
	However, the deterministic data transmission approaches in \cite{leong2018transmission, wang2022transmission, leong2019information, huang2021encryption} are not suitable for scenarios where private and public processes evolve through different processes, as in smart building \cite{li2018information, jia2017privacy}. Moreover, when the remote estimator also acts as a malicious adversary, their assumptions and solutions become inapplicable. To address these challenges, this work investigates stochastic sampling and reconstruction designs for nonlinear stochastic systems with public states and private inputs. We consider a setting where the estimator may infer private inputs from transmitted data, which is suitable for a wide range of scenarios such as privacy-aware sampling and cloud-based feedback control. For general nonlinear systems, we show that the optimal sampling strategy dynamically regulates the adversary’s belief through a dynamic programming decomposition framework. As for linear systems, we derive simplified optimality equations with analytical loss functions and demonstrate that the adversary's estimation covariance of the past private trajectory is regulated by the sampler. Additionally, we propose a stochastic gradient algorithm for co-design optimization, which is validated through simulations. 
	\subsection{Outline}
	The remainder of this paper is organized as follows. Section \ref{Sec:Prob} introduces the system model, provides motivating examples, and formulates the optimization problem for privacy-aware sampling design. Section \ref{Sec:GenStru} derives the optimality equations and structural properties of the state reconstruction policy and the sampler for general nonlinear stochastic systems. Section \ref{Sec:LinStru} presents a simplified design tailored for linear systems. Section \ref{Sec:Alg} develops a policy gradient algorithm to solve the optimization problem for general systems. Simulation results demonstrating the effectiveness of the privacy-aware sampler are provided in Section \ref{Sec:Sim}, followed by conclusions in Section \ref{Sec:Con}.
	\subsection{Notation}
	We use uppercase letters, \emph{e.g.,} $X, Y$, to denote random variables, and their realizations are represented by lowercase letters, \emph{e.g.,} $x, y$. The shorthand $X^K$ represents the sequence $\left[X_0, X_1, \cdots, X_K \right]$. The symbols $p\paren{X}$ and $P\paren{X}$ denote probability density and probability mass functions, respectively. Similarly, $p\paren{\left. X \right| Y}$ and $P\paren{\left. X \right| Y}$ denote the conditional probability density and conditional probability mass functions, respectively. For Gaussian distributions, we use the shorthand $p\paren{X} = \mathcal{N}\paren{X; \tilde{X}, \Sigma}$, which expands to $p\paren{X} \!=\! \frac{1}{\sqrt{\left| 2\pi \Sigma \right|}} \exp\left[ -\frac{1}{2} \paren{X - \tilde{X}}^{\top} \Sigma^{-1} \paren{X - \tilde{X}} \right]$, where $\left| \cdot \right|$ represents the determinant, and $\paren{X - \tilde{X}}^{\top}$ denotes the transpose of $\paren{X - \tilde{X}}$. Furthermore, the entropy of a random variable $X$ is denoted as $H\paren{X} = -\int_{x} p\paren{x} \log p\paren{x} \, dx$. The conditional entropy of $X$ given $Y$ is expressed as $H\paren{\left. X \right| Y} = -\int_{x, y} p\paren{x, y} \log p\paren{\left. x \right| y} \, dx \, dy$. The mutual information between $X$ and $Y$ is given by $I\paren{X; Y} = H\paren{X} - h\paren{\left. X \right| Y}$. Finally, the Kullback–Leibler divergence (KL-divergence) between distributions $p\paren{\cdot}$ and $q\paren{\cdot}$ is defined as $D_{KL}\left[p\paren{X} \| q\paren{X}\right] = \int_{x} p\paren{x} \log \frac{p\paren{x}}{q\paren{x}} \, dx$.
	\section{Problem Formulation}\label{Sec:Prob}
	\subsection{System Model}
	We consider a stochastic system with the state $X_k \in \mathcal{X}\subseteq\R^{n_x}$ and the private input $Y_k\in \mathcal{Y} \subseteq\R^{n_y}$. The system state $X_k$ evolves according to
	\begin{align}\label{Eq.Dynamics-state}
			X_{k+1}\sim p(\left. \cdot\right |X_k,Y_k).
	\end{align}
	The private input $\left\{Y_k\right\}_k$ is a first-order Markov process with the stochastic kernel $p\paren{\left. Y_{k+1} \right|Y_k}$. Using this model, we can study the privacy-aware sampler for both nonlinear systems, 
	\begin{align}
		X_{k+1}=g\left( X_k,Y_k,W_k \right)
	\end{align}
	and linear systems
	\begin{align}
		X_{k+1}=AX_k+BY_k+W_k
	\end{align}
	with the Gaussian or non-Gaussian noise, $W_k$.
	
	We assume that the system state will be resampled and shared with an adversary, as shown in Fig. \ref{Fig.EstInfSys}, where the adversary might infer the private input $Y_k$ based on the sampled data $Z_k$. We next study the optimal design of the stochastic sampler under a privacy requirement.
	\subsection{Motivating Examples}
	In this subsection, we consider two motivating applications of the privacy-aware sampler.
	\subsubsection{Heating, Ventilation, and Air Conditioning (HVAC) System} Consider the HVAC system depicted in Fig. \ref{Fig.NetworkedControl}, where the indoor $CO_2$ concentration, denoted as $X_k$, , is influenced by the time occupants spend inside and outside the room, represented as $Y_k$. To regulate the $CO_2$ level $X_k$, this information is transmitted to an untrusted remote controller, which computes the ventilation speed accordingly \cite{jia2017privacy}. Motivated by this, the sampler would be designed to meet the following objectives: $\left(\romannumeral1\right)$ reduce communication and computation overhead, $\left(\romannumeral2\right)$ maintain a low reconstruction error for $CO_2$ data, and, and $\left(\romannumeral3\right)$ ensure the adversary cannot accurately infer occupancy status from the transmitted data. The central research problem is the design of an effective sampling mechanism to achieve these objectives.
		\begin{figure}[H]
			\centering
			\includegraphics[width=2.4in]{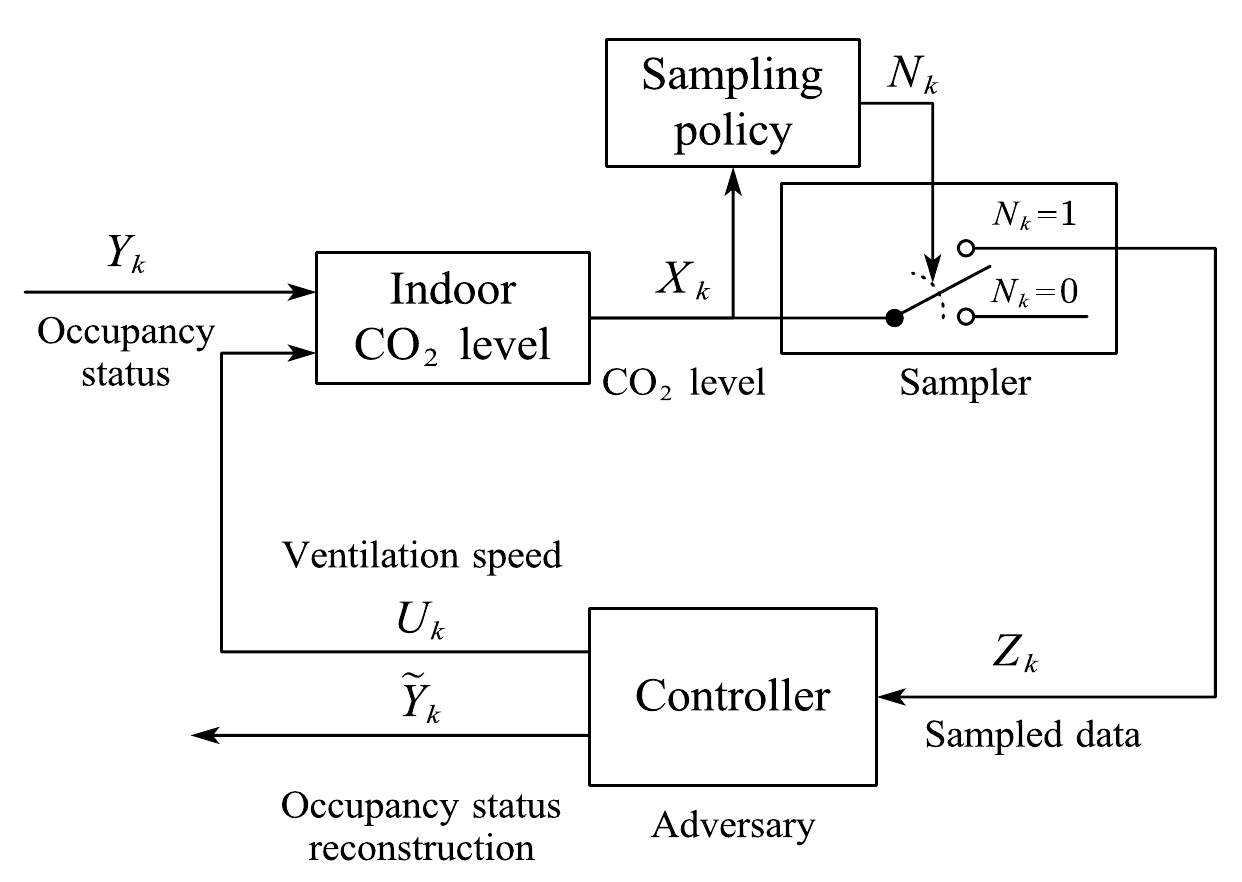}
			\caption{The privacy-aware networked control system. }
			\label{Fig.NetworkedControl}
	\end{figure}
	\subsubsection{Privacy-Aware Causal Compression for Mobility Datasets}
		In location-based services, users' GPS trajectories, denoted as $X^K$, are frequently collected for applications such as traffic analysis and urban planning. These datasets are typically large in size, and there is a need to reduce their volume for efficient transmission and storage. Additionally, such datasets may expose sensitive information, as the exact sequences of locations visited, represented as $Y^K$, can often be uniquely identified \cite{wu2016did}. To mitigate these concerns, a sampling mechanism can be employed before the publication of mobility datasets, with the following objectives: $\left(\romannumeral1\right)$ reduce the data size and storage requirements, $\left(\romannumeral2\right)$ maintain an acceptable level of reconstruction accuracy for the GPS trajectories, and $\left(\romannumeral3\right)$ prevent adversaries from accurately inferring the sequences of visited locations.
	\subsection{Optimal Privacy-Aware Sampler Design}\label{Sec:OPDefinition}
	Consider a sampler with access to the system state $X_k$, as illustrated in Fig. \ref{Fig.EstInfSys}. The sampler makes a decision $N_k \in \{0,1\}$ based on its information set, where $N_k = 1$ indicates that $X_k$ is retained, and $N_k = 0$ means the sample is discarded. The information set of the sampling policy is denoted by $\mathcal{I}_k = \{X^{k}, N^{k-1}\}$, where $X^k$ and $N^{k-1}$ represent the sequences $\{X_0, X_1, \dots, X_k\}$ and $\{N_0, N_1, \dots, N_{k-1}\}$, respectively. 
	
	The output of the sampler is defined as:
	\begin{align}
		Z_k = 
		\begin{cases} 
			X_k, & \text{if } N_k = 1, \\ 
			\emptyset, & \text{if } N_k = 0,
		\end{cases}
		\nonumber
	\end{align}
	where $Z_k = \emptyset$ indicates that $X_k$ has been discarded. We denote the set of discarded samples as $M_k$. For simplicity, we refer to $M_k$ as the \emph{local memory} of the sampler since it contains information not shared with the adversary. Consequently, the information set of the sampling policy can be rewritten as $\mathcal{I}_k = \{Z^{k-1}, X_k, M_k\}$. The sampling policy is represented by $\pi_k(N_k | \mathcal{I}_k)$, which is the conditional probability mass function of $N_k$ given $\mathcal{I}_k$. The action $N_k$ is randomly drawn according to this policy.
	
	Let $\tilde{X}_k$ denote the optimal reconstruction of $X_k$. Under the given sampling policy and observed sequence $Z^k$, the reconstruction satisfies $\tilde{X}_k = Z_k$ if $N_k = 1$. On the other hand, when $N_k = 0$, the adversary reconstructs $X_k$ based on its information set $\tilde{\mathcal{I}}_k = \{Z^k, \tilde{X}^{k-1}\}$, where $\tilde{X}^{k-1}$ denotes the sequence $\{\tilde{X}_0, \tilde{X}_1, \dots, \tilde{X}_{k-1}\}$. The adversary employs a reconstruction policy denoted by $\tilde{\pi}_k(\tilde{X}_k | \tilde{\mathcal{I}}_k) = \tilde{\pi}_k(\tilde{X}_k | Z^k, \tilde{X}^{k-1})$. The reconstruction error is captured by the distortion measure $\textit{l}_D(X_k, \tilde{X}_k)$, where $\textit{l}_D(X_k, \tilde{X}_k) = 0$ if $\tilde{X}_k = X_k$.
	
	With access to $Z^K$, the adversary may attempt to infer the private information $Y^K$, as illustrated in Fig.~\ref{Fig.EstInfSys}. To quantify the information flow from $Y^K$ to $Z^K$, we use the mutual information:
		\begin{equation}
			I(Z^K; Y^K) = \mathsf{E}\left[ \log \frac{p(Y^K | Z^K)}{p(Y^K)} \right].
		\end{equation}
		According to \cite{cover1999elements}, if $I(Z^K; Y^K) = 0$, then $Z^K$ and $Y^K$ are independent. Specifically, if the sampler consistently sends no information to the remote reconstructor, the mutual information becomes zero, as no information about $Y^K$ is revealed. 
	
	On the other hand, since $Y^K \to X^K \to Z^K$ forms a Markov chain, it follows that $I(Z^K; Y^K) \leq I(X^K; Y^K)$. If the sampler always transmits $X_k$, i.e., $Z^K = X^K$, then $I(Z^K; Y^K) = I(X^K; Y^K)$, representing the maximum privacy leakage. Consequently, the mutual information captures the correlation between the transmitted data and the private data. To reduce privacy leakage, the mutual information should be minimized.
	
	Consequently, we formulate the design of the optimal privacy-aware stochastic sampling and reconstruction policies as follows:
		\begin{equation}\label{Eq.OP}
			\min_{\{\pi, \tilde{\pi}\}} L(\pi, \tilde{\pi}) = \min_{\{\pi, \tilde{\pi}\}} \sum_{k=0}^{K} \mathsf{E}[\textit{l}_D(X_k, \tilde{X}_k)] + \lambda I(Z^K; Y^K),
		\end{equation}
		where $\textit{l}_D(X_k, \tilde{X}_k)$ measures the reconstruction error, $I(Z^K; Y^K)$ quantifies the privacy leakage, and $\lambda > 0$ is a penalty constant. The optimization problem in \eqref{Eq.OP} involves two decision variables: the sampling policy $\pi = \{\pi_k\}_{k=0}^K$ and the reconstruction policy $\tilde{\pi} = \{\tilde{\pi}_k\}_{k=0}^K$, which jointly minimize the tradeoff between reconstruction error and privacy leakage.
	\section{Structural Properties of the Optimal Privacy-aware Sampling Design}\label{Sec:GenStru}
	In this section, we first discuss the structure of the optimal reconstruction policy. We then derive the Bellman equations for the optimal stochastic sampler via dynamic programming decomposition. We finally draw a closed-loop control block diagram to interpret that the optimal sampling policy controls the adversary's belief about the private process $Y_k$.
	\subsection{Optimal Reconstruction Policy}
	\begin{figure*}
		\begin{align} \label{Eq.VF}
				\setcounter{equation}{8}
				V_k^{\star}\left( b_k \right) =\min_{\mathcal{A} _k} \tilde{l}_D\left( \mathcal{A} _k,b_k,\tilde{\pi}_{k}^{\star}\left( Z_k=\emptyset,Z^{k-1} \right) \right)  +\lambda l_I\left( \mathcal{A} _k,b_k \right) +\mathsf{E}\left[ V_{k+1}^{\star}\left( \varPhi \left( b_k,\mathcal{A} _k,Z_k \right) \right) |Z^{k-1} \right] ,
		\end{align} 
		\hrule
		\begin{align}\label{Eq.BSUP}
			\setcounter{equation}{11}
			b_{k+1}\left( x_{k+1},m_{k+1},y^{k+1} \right) =\varPhi \left( b_k,\mathcal{A} _k,Z_k \right) =\left\{ \!\! \begin{array}{c}
				\frac{p\left( x_{k+1} \middle| x_k,y_k \right) p\left( y_{k+1} \middle| y_k \right) b_k\left( x_k,m_k,y^k|Z^{k-1} \right) a_k\left( N_k=0 \middle| x_k,m_k \right)}{\int{\int{\int{a_k\left( N_k=0 \middle| x_k,m_k \right) b_k\left( x_k,m_k,y^k|Z^{k-1} \right)}dx_kdm_kdy^k}}},\quad Z_k=\emptyset \\
				\frac{p\left( x_{k+1} \middle| X_k,y_k \right) p\left( y_{k+1} \middle| y_k \right) b_k\left( X_k,m_k,y^k|Z^{k-1} \right) a_k\left( N_k=0 \middle| X_k,m_k \right)}{\int{\int{a_k\left( N_k=0 \middle| X_k,m_k \right) b_k\left( X_k,m_k,y^k|Z^{k-1} \right) dm_kdy^k}}},\quad Z_k=\!X_k\\
			\end{array} \right. \!,
		\end{align}  
		\hrule
		\vspace{-1em}
	\end{figure*}
	We next present the the optimal reconstruction of $X_k$.
	\begin{lemma} \label{Lm.OPEst}
		Given the sampling policy $\pi=\left\{\pi_k\right\}^K_{k=0}$ and the the sampler's output $Z^{k}$, the optimal reconstruction policy is
		\begin{align}\label{Eq.OPTEst}
			\setcounter{equation}{5}
			\tilde{X}_k^{\star}=\tilde{\pi}^{\star}_k\left(Z^{k}\right) = \arg \min_{\tilde{\pi}_k\left(Z^{k}\right)} \mathsf{E}\left[ \left. l_D\left( X_k,\tilde{\pi}_k\left(Z^{k}\right) \right) \right|Z^{k}\right] , 
		\end{align}
		$k=0,1,\cdots ,K.$ If the distortion measure is the squared reconstruction error, \emph{i.e.}, $l_D\left( X_k,\tilde{\pi}_k\left( Z^k \right) \right) =\left( X_k-\tilde{\pi}_k\left( Z^k \right) \right) ^{\top}\left( X_k-\tilde{\pi}_k\left( Z^k \right) \right) $, then the optimal reconstruction is the conditional expectation, $\mathsf{E}\left[ \left. X_k \right|Z^{k} \right] $.
	\end{lemma}
	\begin{proof}
		See Appendix \ref{App.Lm.OPEst}.
	\end{proof}
	According to Lemma \ref{Lm.OPEst}, the optimal reconstruction is solved via a one-step optimization problem. It is important to highlight that the optimal reconstruction policy is dependent on the utilized sampling policy. This is due to the fact that, when $N_k=0$, the expectation in \eqref{Lm.OPEst} is computed using $$p\!\left(\! \left. X_k \right|\!Z^k \!\right) \!=\!\frac{\int{\!\!\pi _k\!\left(\! \left. N_k=0 \right| \!X_k,m_k,Z^{k-1}\! \right)\! p\!\left(\! \left. X_k,m_k \right| \! Z^{k-1}\! \right) \!dm_k}}{P\left( N_k=0 \right)},$$
	which depends on the sampling policy $\pi$. Thus, we can adjust the sampling policy to control the reconstruction error level. 
	\subsection{Optimal Privacy-Aware Sampler}\label{Sec:OPET}
	In this subsection, we study the optimal sampler in a dynamic decomposition manner. We first show the optimization problem \eqref{Eq.OP} is equivalent to an auxiliary optimization problem in which the decision variable is a policy collection upon the sampled data. Then, we study the structural properties of the optimal sampler based on its optimality equations. 
	
	Consider the following policy collection,
	\begin{align}
		\mathcal{A}_k\left(Z^{k-1}\right)=\left\{{a_k\paren{\left. n_k\right|x_k,m_k}}, \forall n_k, x_k, m_k\right\}, 
	\end{align}
	where $a_k\paren{\left. N_k\right|X_k,M_k}$ is a conditional mass distribution of $N_k$ dependent on $X_k$ and $M_k$. We next consider the following auxiliary optimization problem for $\mathcal{A}=\left\{\mathcal{A}_k\left( Z^{k-1} \right)\right\}^K_{k=0}$,
	\begin{equation}\label{Eq.OP2}
		\setcounter{equation}{8}
		\min_{\mathcal{A}} L\!\left(\mathcal{A}\right)\!=\!\min_{\mathcal{A}} \sum_{k=0}^{K}\!\mathsf{E} \left[\textit{l}_{D}\!\left(X_k,\tilde{\pi}^{\star}_k\!\left(Z^{k}\right)\right)\!\right]+\lambda I\! \left(Z^{K};\!Y^{K}\right)\!,\!
	\end{equation}
	where given $\paren{Z^{k-1}, X_k, M_k}$, we first select the policy collection $\mathcal{A}_k$ based on $Z^{k-1}$, and then select $a_k\paren{\left. N_k\right|X_k,M_k}$ based on $\paren{X_k, M_k}$, and finally generate $N_k$ by sampling it from $a_k\paren{N_k \left| X_k,M_k \right.}\in \mathcal{A}_k$.  Without loss of generality, we can compute the optimal policy using the auxiliary problem.
	\begin{lemma} \label{Lm.Equivalent}
		The original and auxiliary optimization problems, \eqref{Eq.OP} and \eqref{Eq.OP2} are equivalent.
	\end{lemma}
	\begin{proof}
		See Appendix \ref{App.Lm.Equivalent}.
	\end{proof}
	According to Lemma \ref{Lm.Equivalent}, the optimal policy derived from the auxiliary optimization problem \eqref{Eq.OP2} is also optimal for \eqref{Eq.OP}. In the next theorem, we derive its Bellman optimality equations.
	\begin{theorem} \label{Th.OPTEQU}
		Let $V^\star _k\left(\cdot\right)$ denote the optimal cost-to-go function of the optimization problem \eqref{Eq.OP2} at time $k$. Then $V^\star _k\left(\cdot\right)$ can be recursively computed using \eqref{Eq.VF}, where
		\begin{align}\label{Eq.DisLoss}
			\setcounter{equation}{9}
			&\tilde{l}_D\!\left( \mathcal{A} _k,b_k, \tilde{\pi}_{k}^{\star}\left( Z_k\!=\!\emptyset,Z^{k-1} \right) \right) \nonumber \\=&\mathsf{E}\left[ l_D\left( X_k,\tilde{\pi}_{k}^{\star}\left( Z_k=\emptyset ,Z^{k-1} \right) \right) \middle| Z^{k-1} \right] , 
		\end{align}
		\begin{align}\label{Eq.InfoLoss}
			l_I\left( \mathcal{A} _k,b_k \right) =\mathsf{E}\left[ \log \frac{p\left( Y^k \middle| Z^k \right)}{p\left( Y^k \middle| Z^{k-1} \right)} \middle| Z^{k-1} \right] ,
		\end{align}
		$b_k$ is the belief state, $b_k\left( x_k,m_k,y^{k} \right) =p\left( x_k,m_k,y^{k}|Z^{k-1} \right) $, and  $ V_{K+1}^{\star}\left(\cdot\right)$ is zero. The belief state $b_{k}$ can be recursively computed via \eqref{Eq.BSUP}, and $b_0$ is the initial joint probability density function of $X_0$ and $Y_0$, \emph{i.e.}, $b_0=p_{x_0, y_0}$. 
	\end{theorem}
	\begin{proof}
		See Appendix \ref{App.Th.OPTEQU}.
	\end{proof}
	According to Theorem \ref{Th.OPTEQU}, the optimal solution of  \eqref{Eq.VF} is a function of the belief state $b_k$, \emph{i.e.,} $\mathcal{A}^{\star}_k\left(b_k\right)$, since the optimality equation \eqref{Eq.VF} depends on $b_k$. Based on Lemma \ref{Lm.Equivalent} and Theorem \ref{Th.OPTEQU}, the optimal sampling policy can be obtained via first updating the belief state based on the shared information $Z^{k-1}$, then optimizing the policy collection $\mathcal{A}_k\left(b_k\right)$ in \eqref{Eq.VF}, and finally selecting the conditional mass distribution $a_k\paren{\left. N_k\right|X_k,M_k}$ according to the sampler's local information $X_k$ and $M_k$. We summarize this dynamic decomposition method in the next theorem.
	\begin{theorem}\label{Th.ETPolicy}
		Let $\mathcal{A}^{\star}_k\left(b_k\right)$ denote the solution to \eqref{Eq.VF}. Then, given $\mathcal{I} _k = \left\{Z^{k-1},X_k,M_k\right\}$, the optimal sampling policy at time $k$ is $\pi_k^{\star}\left(\tilde{X}_k\left|X_k,M_k,b_k\right.\right)=a_k^{\star}\left(\tilde{X}_k\left|X_k,M_k\right.\right)\in\mathcal{A}^{\star}_k\paren{b_k}$, where $b_k$ is the belief state associated with $Z^{k-1}$. 
	\end{theorem}
	\begin{figure*}
		\begin{align}\label{Eq.ETExp}
			\setcounter{equation}{15}
			\pi _k\left( \left. N_k \right|X_k,Z^{k-1} \right) =\left\{ \begin{array}{c}
				\exp \left( -\frac{1}{2}\left( X_k-g_k\left( Z^{k-1} \right) \right) ^{\top}f_{k}^{-1}\left( Z^{k-1} \right) \left( X_k-g_k\left( Z^{k-1} \right) \right) \right) ,\quad N_k=0\\
				1-\exp \left( -\frac{1}{2}\left( X_k-g_k\left( Z^{k-1} \right) \right) ^{\top}f_{k}^{-1}\left( Z^{k-1} \right) \left( X_k-g_k\left( Z^{k-1} \right) \right) \right) ,N_k=1\\
			\end{array} \right. .
		\end{align}
		\hrule
		\vspace{-1em}
	\end{figure*}
	We show the structure of the optimal sampler in Fig. \ref{Fig.EVStru} based on Theorem \ref{Th.ETPolicy}. As shown in this figure, the update of belief state is controlled by the sampling policy via the belief state update equation  $b_{k+1}=\varPhi\left(b_k,\mathcal{A}_k,Z_k\right)$. Since $b_{k+1}$ is the adversary's belief of the private information $Y^{k+1}$ given the sampler's output $Z^k$, the privacy-aware sampling design is a closed-loop control problem where the sampler controls the belief of the adversary about the private information.
	\begin{figure}[H]
		\centering
		\includegraphics[width=2.4in]{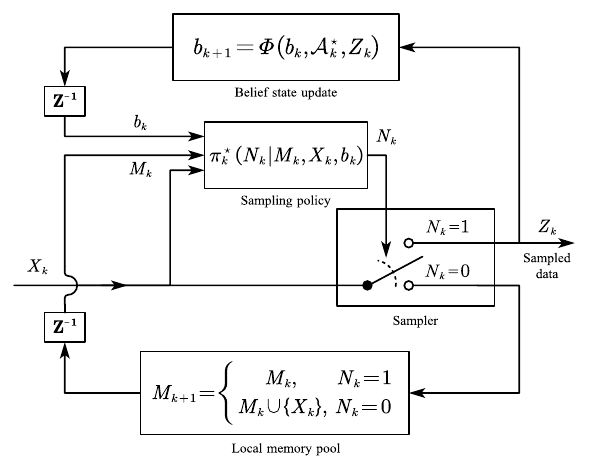}
		\caption{The structure of optimal sampler. }
		\label{Fig.EVStru}
	\end{figure}
	\section{Application to Linear Gaussian Systems}\label{Sec:LinStru}
	In this section, we address the design of a privacy-aware sampler when $X_k$ and $Y_k$ represent the states of a linear Gaussian system. We propose a low-complexity solution to achieve an optimal trade-off between the reconstruction error of $X_k$ and the privacy leakage associated with $Y_k$. Specifically, consider the following system dynamics:
	\begin{align}\label{Eq.LinSysGau}
		\setcounter{equation}{12}
		\left[ \begin{array}{c}
			X_{k+1}\\
			Y_{k+1}
		\end{array} \right] = A \left[ \begin{array}{c}
			X_k\\
			Y_k
		\end{array} \right] + \left[ \begin{array}{c}
			W_{x,k}\\
			W_{y,k}
		\end{array} \right],
	\end{align}
	where $X_k \in \mathbb{R}^{n_x}$ denotes the observable state of the system, which can be sampled, and $Y_k \in \mathbb{R}^{n_y}$ represents the private state of the system. The noise terms $W_{x,k}$ and $W_{y,k}$ are independent and identically distributed (\emph{i.i.d.}) Gaussian random vectors with zero mean and a covariance matrix given by
	\small
	\begin{align}\label{Eq.LinGauSys}
		Q_{k,j} = \mathsf{E} \left[\! \left[ \begin{array}{c}
			\!W_{x,k}\!\\
			\!W_{y,k}\!
		\end{array} \right] 
		\left[ \begin{array}{c}
			\!W_{x,j}\!\\
			\!W_{y,j}\!
		\end{array} \right]^\top \!\right] = 
		\left[ \begin{matrix}
			Q^{xx} & Q^{xy}\\
			Q^{yx} & Q^{yy}
		\end{matrix} \right] \delta[k-j], \!
	\end{align}
	\normalsize
	where $\delta[k-j]$ is the Kronecker delta function, such that $\delta[k-j] = 1$ if $k = j$ and $\delta[k-j] = 0$ otherwise. For notational convenience, we use $Q$ to denote $Q_{k,j}$ when $k = j$. 
	
	Additionally, the initial states $X_0$ and $Y_0$ are Gaussian random variables with constant mean values $\tilde{x}_{0|-1}$ and $\tilde{y}_{0|-1}$, respectively. The covariance matrix of the joint distribution of $X_0$ and $Y_0$ is denoted as $P_{0|-1}$. If the evolution of $Y_{k+1}$ is independent of $X_k$, this model reduces to a special case of the general system dynamics given by Equation~\eqref{Eq.Dynamics-state}.
	
	The optimality equation \eqref{Eq.VF} presented in the previous section offers a theoretical framework for designing the optimal privacy-aware sampler for the linear system in \eqref{Eq.LinSysGau}. However, the implementation of this optimization framework is challenging, even for the simple linear system, due to the computational complexity involved in evaluating mutual information and updating the belief state. We next show that, by leveraging the structure of the linear Gaussian system, both the loss computation and the belief state updates can be significantly simplified. For convenience, we consider the following optimization problem similar to \eqref{Eq.OP},
		\begin{align}\label{Eq.GauObj}
			\setcounter{equation}{14}
			\min_{\left\{\pi, \tilde{\pi}\right\}}  \sum_{k=0}^{K}\mathsf{E}\left[\left(\!X_k-\tilde{X}_k\!\right)^{\top}\!\left(\!X_k-\tilde{X}_k\!\right)\!\right]+\lambda I\!\left( \! Z^K;Y^K \!\right).
	\end{align}
	\subsection{Stochastic Sampling Policy Design}
	Since the sampling policy is dependent on the system states $X^k$ as well as its past decisions $N^{k-1}$, the conditional distributions of the states of the system given the sampler's output may not be Gaussian. To solve this problem, we next propose a stochastic sampling policy which ensures that these distributions remains Gaussian. 
	
	To this end, let $\epsilon_k\!\sim\! \mathbf{U}\!\left[ 0,1 \right]$ be a random variable sampled from the uniform distribution, $f_k\!\left(Z^{k-1}\right)$ be a positive semi-definite symmetric matrix function and $g_k\!\left(Z^{k-1}\right)$ be a mapping from $Z^{k-1}$ to $\mathrm{R}^{n_x}$. The sampler decides to discard the sample, \emph{i.e.}, $N_k=0$, if $\epsilon_k \leq \exp \! \left( -\frac{1}{2}\left( X_k-g_k\left( Z^{k-1} \right) \right) ^{\top}\!f_{k}^{-1}\left( Z^{k-1} \right) \left( X_k-g_k\left( Z^{k-1} \right) \right) \right)$. Otherwise, $N_k=1$, and $X_k$ will be kept. This stochastic policy is expressed in the form of \eqref{Eq.ETExp}.
	\begin{remark}
		In the context of event-triggered control \cite{han2015stochastic, demirel2018tradeoffs}, stochastic triggering policies similar to \eqref{Eq.ETExp} have been proposed to enable the application of Kalman filtering. However, unlike \cite{han2015stochastic, demirel2018tradeoffs}, we design a privacy-aware sampling policy in the form of \eqref{Eq.ETExp} by solving the privacy-utility trade-off optimization problem \eqref{Eq.GauObj}.
	\end{remark}
	\subsection{Optimal Reconstruction Policy}
	In this subsection, we show that the optimal reconstruction policy can be analytically derived based on the conditional distributions of system states. To this end, we first derive recursive update equations for these distributions.
	\begin{lemma}\label{Lm.GaussianBelief}
		The following statements are true under the stochastic sampling policy in \eqref{Eq.ETExp}.
		\begin{itemize}
			\item The belief state $p\left(\left. X_k, Y^{k} \right|Z^{k-1}\right)$ is a conditional Gaussian distribution with the mean $\left(\tilde{x}_{k|k-1}, \tilde{y}^{k|k-1}\right)$ and covariance matrix $P_{k|k-1}$.
			\item The conditional distribution $p\!\left(\!\left. X_k, \!Y^k \right|\!Z^{k}\!\right)$ is Gaussian, with the mean $\!\left(\!\tilde{x}_{k|k}, \!\tilde{y}^{k|k}\!\right)\!$ and covariance matrix $P_{k|k}$.
			\item $\left(\!\tilde{x}_{k|k-1}, \tilde{y}^{k|k-1}\!\right)$ and $\tilde{P}_{k|k-1}$ are recursively computed via,
			\begin{align}
				\left[ \begin{array}{c}
					\tilde{x}_{k|k-1}\\
					\tilde{y}^{k|k-1}\\
				\end{array} \right] =A_k\left[ \begin{array}{c}
					\tilde{x}_{k-1|k-1}\\
					\tilde{y}^{k-1|k-1}\\
				\end{array} \right] , \nonumber\\
				P_{k|k-1}=A_kP_{k-1|k-1}A_{k}^{\top}+\Sigma _k, \nonumber
			\end{align}
			where
			\begin{align}
				\!\!A_k=\!\!\left[ \begin{matrix}
					A&		\!\!\!\!0_{n\times n_y\left(k-1\right)}\\
					0_{kn_y\times n_x}&		\!\!\!\!I_{kn_y\times kn_y}\\
				\end{matrix} \right] \!,  \Sigma _k=\!\!\left[ \begin{matrix}
					Q&	\!\!	0_{n\times kn_y}\\
					0_{kn_y\times n}&	\!\!	0_{kn_y\times kn_y}\\
				\end{matrix} \!\right] \!. \nonumber
			\end{align}
			Given the initial state distribution of $X_0$ and $Y_0$, \emph{i.e.}, $p\left(X_0,Y_0\right)$, $\tilde{x}_{0|-1}$ and $\tilde{y}_{0|-1}$ are the mean of $X_0$ and $Y_0$, respectively. $P_{0|-1}$ is the covariance matrix of $p\left(X_0,Y_0\right)$.
			\item $\tilde{x}_{k|k}$, $\tilde{y}^{k|k}$ and $\tilde{P}_{k|k}$ are computed as follows.
			If $N_k$ is zero, then we have
			\begin{align}
				\!\!\left[\!\! \begin{array}{c}
					\tilde{x}_{k|k}\\
					\tilde{y}^{k|k}\\
				\end{array} \!\!\right] \!=\!\left( D_k+I \right) ^{-1}&\!\left(\! D_k\!\left[\!\! \begin{array}{c}
					g_k\left( Z^{k-1} \right)\\
					0_{n_y\left( k+1 \right) \times 1}\\
				\end{array} \!\!\right] \!+\!\left[\!\! \begin{array}{c}
					\tilde{x}_{k|k-1}\\
					\tilde{y}^{k|k-1}\\
				\end{array} \!\!\right] \right) \!,\nonumber\\
				P_{k|k}=&\left( D_k+I \right) ^{-1}P_{k|k-1},\nonumber
			\end{align}
			where
			\begin{align}
				D_k=P_{k|k-1}\left[ \begin{matrix}
					f_{k}^{-1}\left( Z^{k-1} \right)&		0_{n_x\times n_y\left( k+1 \right)}\\
					0_{n_y\left( k+1 \right) \times n_x}&		0_{n_y\left( k+1 \right) \times n_y\left( k+1 \right)}\\
				\end{matrix} \right] . \nonumber
			\end{align}
			If $N_k$ is one, then we have $\tilde{x}_{k|k}=Z_k$ and
			\begin{align}
				\tilde{y}^{k|k}=\tilde{y}^{k|k-1}&+P_{k|k-1}^{xy}\left( P_{k|k-1}^{xx} \right) ^{-1}\left( Z_k-\tilde{x}_{k|k-1} \right) ,\nonumber\\
				P_{k|k}=&\left[ \begin{matrix}
					0_{n_x\times n_x}&		0_{n_x\times n_y\left( k+1 \right)}\\
					0_{n_y\left( k+1 \right) \times n_x}&		P_{k|k}^{yy}\\
				\end{matrix} \right] ,\nonumber
			\end{align}
			where
			\begin{align}
				P_{k|k}^{yy}=P_{k|k-1}^{yy}-P_{k|k-1}^{yx}\left( P_{k|k-1}^{xx} \right) ^{-1}P_{k|k-1}^{xy}. \nonumber
			\end{align}
		\end{itemize}
	\end{lemma}
	\begin{proof}
		See Appendix \ref{App.Lm.GaussianBelief}.
	\end{proof}
	According to Lemma \ref{Lm.GaussianBelief}, given $Z^{k}$, the adversary's belief about $X_k$ and $Y^k$ are fully captured by the conditional mean $\tilde{x}_{k|k}$, $\tilde{y}^{k|k}$ and covariance matrix $P_{k|k}$. We next derive the optimal reconstruction policy based on Lemma \ref{Lm.GaussianBelief}.
	\begin{lemma}\label{Lm.OptLinState}
		Given the sampling policy collection \eqref{Eq.ETExp} with $f_k$ and $g_k$, and the sampler's output $Z^{k}$,	the optimal reconstruction of $X_k$ is given by $\tilde{X}_k=\tilde{x}_{k|k}$ as stated in Lemma \ref{Lm.GaussianBelief}.
	\end{lemma}
	\begin{proof}
		According to Lemma \ref{Lm.OPEst}, the optimal reconstruction policy is the conditional expectation $\mathsf{E}\left[ \left. X_k \right|Z^{k} \right] $. Since the conditional distribution $p\left(\left. X_k, Y^{k} \right|Z^{k}\right)$ is Gaussian with the mean $\tilde{x}_{k|k}$ as shown in Lemma \ref{Lm.GaussianBelief}, the optimal reconstruction of $X_k$ is $\tilde{x}_{k|k}$.
	\end{proof}
	So far, we have shown that the belief state of sampler is Gaussian, and can be updated in a recursive manner. These properties will play essential roles in reducing the optimization complexities of the privacy-aware sampler. 
	\begin{figure*}
		\begin{align}\label{Eq.OneStepLoss}
			\setcounter{equation}{17}
			l_k\!\left( f_k,g_k,Z^{k-1} \right) =&P\left( N_k=0 \middle| Z^{k-1} \right) f_k\left( Z^{k-1} \right) \left( f_k\left( Z^{k-1} \right) +P_{k|k-1}^{xx} \right) ^{-1}P_{k|k-1}^{xx} +\lambda \log \sqrt{\left| P_{k|k-1}^{yy} \right|} \nonumber
			\\
			&-\lambda \left( 1-P\left( N_k=0 \middle| Z^{k-1} \right) \right) \left( \log \sqrt{\left|  \left( P_{k|k-1}^{yy}-P_{k|k-1}^{yx}\left( P_{k|k-1}^{xx} \right) ^{-1}P_{k|k-1}^{xy} \right) \right|} \right) \nonumber
			\\
			&-\lambda P\left( N_k=0 \middle| Z^{k-1} \right) \left( \log \sqrt{\left|  \left( P_{k|k-1}^{yy}-P_{k|k-1}^{yx}\left( f_k\left( Z^{k-1} \right) +P_{k|k-1}^{xx} \right) ^{-1}P_{k|k-1}^{xy} \right) \right|} \right),
		\end{align}
		\hrule
		\begin{align}
			\!\!\!P\!\left(\! N_k\!=\!0 \middle| Z^{k-1} \!\right) \!=\!\!\!\sqrt{\!\frac{\left| f_k\!\left( Z^{k-1} \right) \right|}{\left|\! \left(\! f_k\!\left( Z^{k-1} \right) \!+\!P_{k|k-1}^{xx} \right) \!\right|\!}}\exp\! \left( \!\!-\frac{1}{2}\!\left(\! g_k\!\left( Z^{k-1} \right) \!-\!\tilde{x}_{k|k-1} \!\right) ^{\top}\!\!\!\left(\! f_k\!\left( Z^{k-1} \right)\! +\!P_{k|k-1}^{xx} \!\right) ^{-1}\!\!\left(\! g_k\!\left(\! Z^{k-1} \right) \!-\!\tilde{x}_{k|k-1} \!\right) \!\!\right) \!, \!\!
		\end{align}
		\hrule
		\begin{align} \label{Eq.LinOptEq}
			V_{k}^{\star}\left( \tilde{x}_{k|k-1},P_{k|k-1} \right) =\min_{f_k, g_k} l_k\left( f_k, g_k,\tilde{x}_{k|k-1},P_{k|k-1} \right) +\mathsf{E}\left[ V_{k+1}^{\star}\left( \tilde{x}_{k+1|k},P_{k+1|k} \right) \middle| Z^{k-1} \right].
		\end{align}
		\hrule
		\vspace{-1em}
	\end{figure*}
	\subsection{Optimal Privacy-Aware Sampler}
	In this subsection, we study the structural properties of the optimal sampling policy. To this end, we first show that the optimization objective function \eqref{Eq.GauObj} can be analytically expressed with the conditional mean and covariance matrix shown in Lemma \ref{Lm.GaussianBelief}. 
	\begin{lemma}\label{Lm.GaussianOpt}
		The optimization problem \eqref{Eq.GauObj} can be written as,
		\begin{align}\label{Eq.OPLINEAR}
			\setcounter{equation}{16}
			\min_{\left\{ f_k, g_k \right\} _{k=0}^{K}}\!\! L\left( f,g \right) =\!\min_{\left\{ f_k, g_k \right\} _{k=0}^{K}}\! \mathsf{E}\!\left[ \sum_{k=0}^K{l_k\!\left( f_k, g_k,Z^{k-1} \right)} \!\right] ,\!
		\end{align}
		with the one-step loss $l_k\left( f_k,g_k,Z^{k-1} \right)$ in \eqref{Eq.OneStepLoss}.
	\end{lemma}
	\begin{proof}
		See Appendix \ref{App.Lm.GaussianOpt}.
	\end{proof}
	As shown in Lemma \ref{Lm.GaussianOpt}, the one-step loss \eqref{Eq.OneStepLoss} has four terms, where the first term $f_k\!\left(\! Z^{k-1} \!\right)\! \left(\! f_k\!\left(\! Z^{k-1} \!\right) +P_{k|k-1}^{xx} \!\right) ^{-1}P_{k|k-1}^{xx}$ is the conditional average reconstruction error when $N_k$ is zero, the second term $-\log \sqrt{\left| P_{k|k-1}^{yy} \right|}$ is the inherent uncertainties of the adversary about $Y^k$ before receiving $Z_k$, the third one and the last term are the privacy leakages with $N_k=1$ and $N_k=0$, respectively.
	
	Compared with the auxiliary optimization problem \eqref{Eq.OP2}, the optimization problem for linear systems \eqref{Lm.GaussianOpt} is simplified due to two reasons. First, given the sampling policy, we do not need to further optimize the reconstruction policy, since it can be directly obtained based on the recursive equations in Lemma \ref{Lm.GaussianBelief}. Second, the one-step loss function is analytical, which reduces the computation complexities of information loss in \eqref{Eq.InfoLoss}. 
	\begin{figure*}
		\begin{align}\label{Eq.ETPG}
			\setcounter{equation}{23}
			\nabla _{\theta}L\left( \theta, \varphi^\star\left(\theta\right) \right) &=\mathsf{E}\left\{ \sum_{k=0}^K{\left( \mathrm{J}_{\tilde{\pi}_{k,\varphi^\star\left(\theta\right)}}\left( \theta \right)  \right) ^{\top}\nabla _{\tilde{\pi}_{k,\varphi^\star\left(\theta\right)}}l_D\left( X_k,\tilde{\pi}_{k,\varphi^\star\left(\theta\right)} \right)} \right\} \nonumber
			\\
			&+\mathsf{E}\left[ \left( \sum_{k=0}^K{l_D\left( X_k,\tilde{\pi}_{k,\varphi^\star\left(\theta\right)} \right) +\lambda \log \frac{p_{\theta}\left( Y^k \mid Z^k \right)}{p_{\theta}\left( Y^k \mid Z^{k-1} \right)}} \right) \left( \sum_{k=0}^K{\nabla _{\theta}\log \pi _{\theta}\left( N_k \mid X^k,N^{k-1} \right)} \right) \right] ,
		\end{align}
		\hrule
		\begin{align}\label{Eq.JacMatTheta}
			\setcounter{equation}{25}
			\mathrm{J}_{\varphi^\star\left(\theta\right)}\left( \theta \right) =-\left(\! \nabla _{\varphi}^{2}\mathsf{E}\!\left[ \sum_{k=0}^K{l_D\left( X_k,\tilde{\pi}_{k,\varphi^\star\left(\theta\right)} \right)} \right] \right) ^{-1}\!\mathsf{E}\!\left[ \left( \sum_{k=0}^K{\nabla _{\varphi}l_D\left( X_k,\tilde{\pi}_{k,\varphi^\star\left(\theta\right)} \right)}\! \right)\! \otimes \!\left( \sum_{k=0}^K{\nabla _{\theta}\log \pi _{\theta}\left( N_k \middle| X^k,N^{k-1} \right)} \!\right) \!\right] \!.\!
		\end{align}
		\hrule
		\vspace{-1em}
	\end{figure*}
	
	We next derive the dynamic programming decomposition optimization formulation for \eqref{Eq.OPLINEAR}.
	\begin{theorem}\label{Th.OPTEQULIN}
		The optimal sampling policy collection can be obtained by solving the recursive optimality equation  \eqref{Eq.LinOptEq}, where the optimal $f_k, g_k$ are functions of $\tilde{x}_{k|k-1}$ and $P_{k|k-1}$, \emph{i.e.}, $f_k^\star\left(\tilde{x}_{k|k-1}, P_{k|k-1}\right)$ and $g_k^\star\left(\tilde{x}_{k|k-1}, P_{k|k-1}\right)$, and $\tilde{x}_{k|k-1}$ and $P_{k|k-1}$ can be recursively computed using Lemma \ref{Lm.GaussianBelief}.
	\end{theorem}
	\begin{proof}
		The optimal cost-to-go function can be derived in a backward manner starting from $V_{K+1}^\star\left(Z^K\right)=0$ similar to Theorem \ref{Th.OPTEQU}. Since $Z^{k-1}$ can be compressed as the conditional mean $\tilde{x}_{k|k-1}$ and variance $P_{k|k-1}$ as shown in Lemma \ref{Lm.GaussianOpt}, the optimal cost-to-go function, $f_k$ and $g_k$ are functions of $\tilde{x}_{k|k-1}$ and $P_{k|k-1}$.
	\end{proof}
	According to Theorem \ref{Th.OPTEQULIN}, to obtain the optimal sampling policy, we should first optimize $f_k$ and $g_k$, and then construct the sampling policy \eqref{Eq.ETExp} with the system state $X_k$, which is consistent with the dynamic programming decomposition results for general systems in Theorem \ref{Th.ETPolicy}. 
	
	We illustrate the structure of the optimal sampling policy in the linear case in Fig. \ref{Fig.LinearEVStru}, where the conditional mean $\tilde{x}_{k+1|k}$ and covariance matrix $P_{k+1|k}$ are recursively updated using Lemma \ref{Lm.GaussianBelief}. Compared with the optimal sampling policy in Fig. \ref{Fig.EVStru}, the sampling policy for linear Gaussian systems has a similar structure, except that the local memory pool is removed, since we adopt a special policy collection \eqref{Eq.ETExp} without memory state $M_k$. Since the adversary's reconstruction uncertainty, \emph{i.e.}, the conditional covariance $P_{k|k}^{yy}$, is determined by the sampling policy $f_k^\star$, $g_k^\star$ and its output $Z_k$, the adversary's belief about privacy is controlled by the sampler.
	\begin{figure}[H]
		\centering
		\includegraphics[width=2.8in]{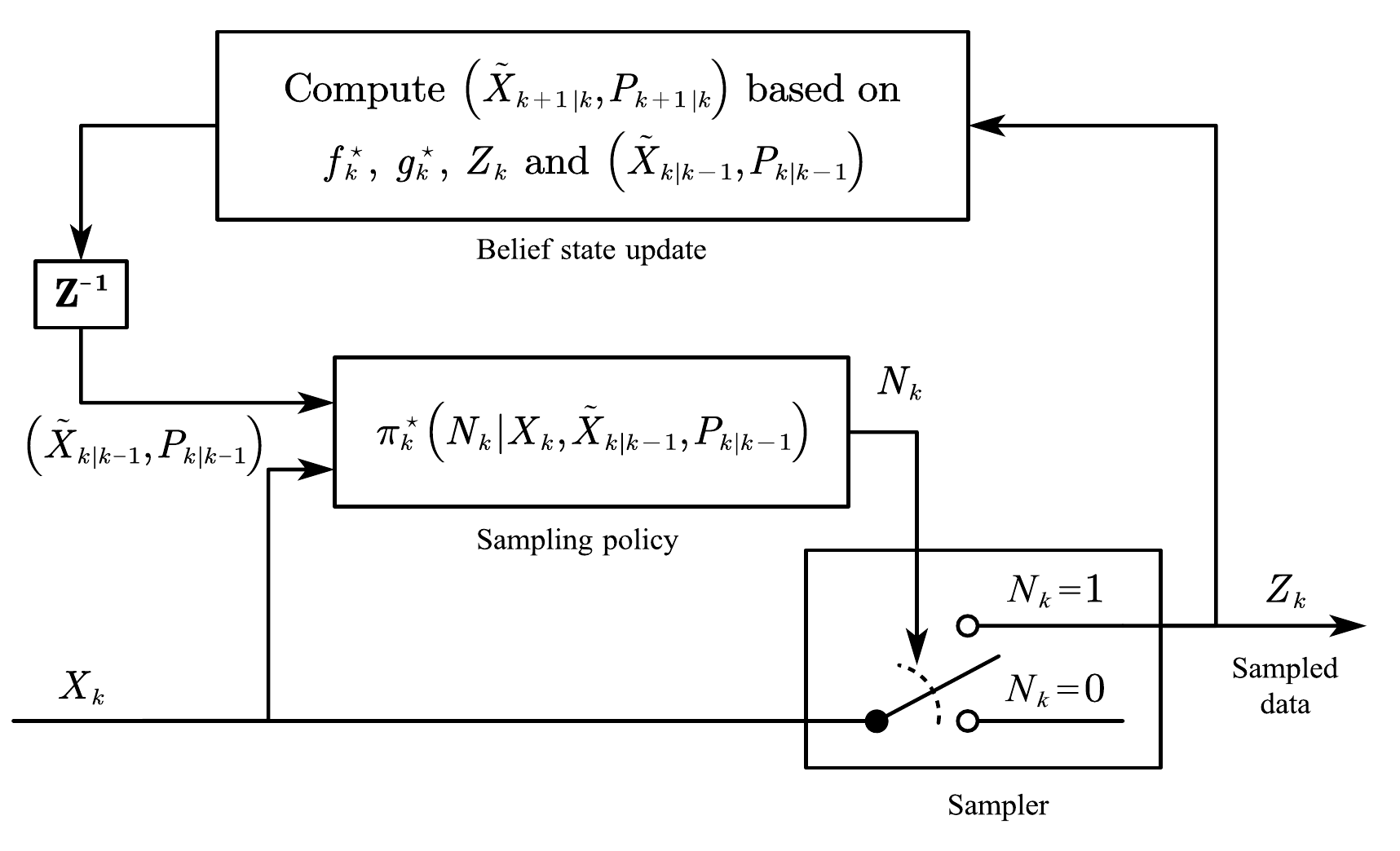}
		\caption{The structure of sampling policy for linear systems. }
		\label{Fig.LinearEVStru}
	\end{figure}
	\section{A Numerical Algorithm for Joint Design of Sampling and Reconstruction Polices}\label{Sec:Alg}
	In this section, we develop a stochastic gradient algorithm for computing the sampling and reconstruction policies. To this end, we first parameterize $\pi$ and $\tilde{\pi}$ with $\varphi$ and $\theta$, respectively. In our numerical algorithm, we consider the following Stackelberg formulation for optimization
	\begin{align} \label{Eq.ETParamOpt}
		\setcounter{equation}{20}
		\!\! \theta ^\star=&\arg\min_{\theta} L\left( \theta ,\varphi^\star\left(\theta\right) \right)\nonumber \\=&\arg\min_{\theta}\! \sum_{k=0}^K{\!\mathsf{E}_{\theta}\!\left[ l_D\left( \!X_k,\tilde{\pi}_{k,\varphi^\star\left(\theta\right)} \right) \right]}\!+\!\lambda I_{\theta}\!\left( Z^K;Y^K \right) \!,\!
	\end{align}
	\begin{align}\label{Eq.EstParamOpt}
		\varphi ^\star\left(\theta\right)=&\arg\min_{\varphi\left(\theta\right)} \tilde{L}\left( \theta ,\varphi\left(\theta\right) \right) \nonumber \\=&\arg\min_{\varphi\left(\theta\right)} \mathsf{E}_{\theta}\left[ \sum_{k=0}^K{l_D\left( X_k,\tilde{\pi}_{k,\varphi\left(\theta\right)} \right)} \right] ,
	\end{align}
	where the sampler acts as the leader and the reconstruction policy acts as the follower. The gradient of the adversary's objective function in \eqref{Eq.EstParamOpt} is given by,
	\begin{align}\label{Eq.GradEst}
		\nabla _{\varphi}\tilde{L}\left( \theta ,\varphi\left(\theta\right) \right) =\mathsf{E}_{\theta}\left[ \sum_{k=0}^K{\nabla _{\varphi}l_D\left( X_k,\tilde{\pi}_{k,\varphi\left(\theta\right)} \right)} \right] ,
	\end{align}
	which is due to the fact that the reconstruction policy does not impact the evolution of system and the sampler with fixed $\theta$.
	
	Since the reconstruction policy acts as the follower, its best response is a function of the sampling policy, \emph{i.e.}, $\varphi^\star\left(\theta\right)$. Thus, the standard policy gradient theorem \cite{bertsekas2019reinforcement} cannot be used to compute the gradient of the objective function in \eqref{Eq.ETParamOpt} with respect to $\theta$. To tackle this challenge, inspired by \cite{fiez2020implicit}, we derive $\nabla _{\theta}L\left( \theta, \varphi^\star\left(\theta\right) \right)$ based on the implicit function theorem and present it in the next lemma.
	\begin{lemma}\label{Lm.PG}
		The gradient of the objective function \eqref{Eq.ETParamOpt} with respect to $\theta$ is given by \eqref{Eq.ETPG}, where 
		\begin{align}
			\setcounter{equation}{24}
			\mathrm{J}_{\tilde{\pi}_{k,\varphi^\star\left(\theta\right)}}\left( \theta \right) =\mathrm{J}_{\tilde{\pi}_{k,\varphi^\star\left(\theta\right)}}\left( \varphi^\star\left(\theta\right) \right) \mathrm{J}_{\varphi^\star\left(\theta\right)}\left( \theta \right),
		\end{align}
		$J$ is the Jacobian matrix, \emph{e.g.}, $\mathrm{J}_{\varphi^\star\left(\theta\right)}\left( \theta \right)$ computed using \eqref{Eq.JacMatTheta}, is the differential of $\varphi^\star\left(\theta\right)$ at $\theta$, $\nabla ^2$ represents the Hessian matrix, $\otimes$ is the outer product, and $\varphi^\star\left(\theta\right)$ satisfies $\nabla _{\varphi}\tilde{L}\left( \theta ,\varphi^\star\left(\theta\right) \right) =0$.
	\end{lemma}
	\begin{proof}
		See Appendix \ref{App.Lm.PG}.
	\end{proof}
	\begin{figure*}
		\center
		\subfigure[]{
			\centering
			\includegraphics[width=2.2in]{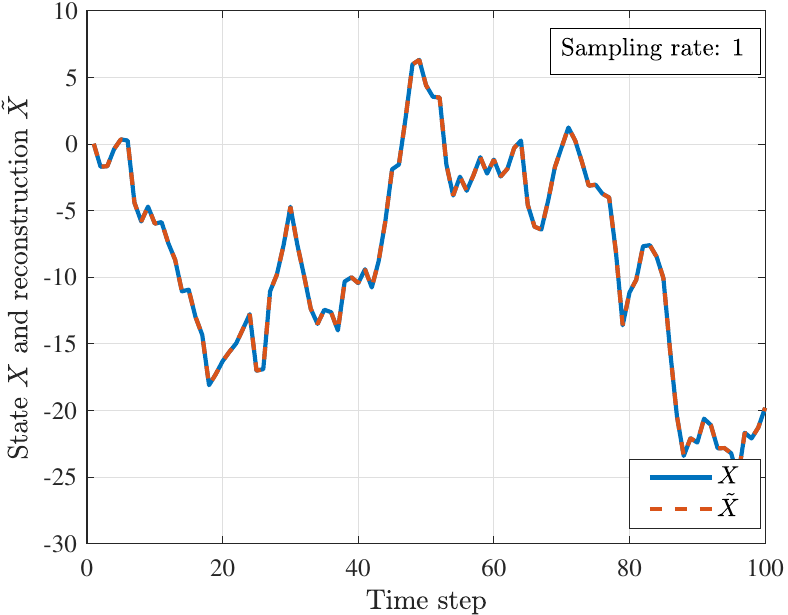}\label{Fig.EstOriX}	
		}	
		\subfigure[]{
			\centering
			\includegraphics[width=2.2in]{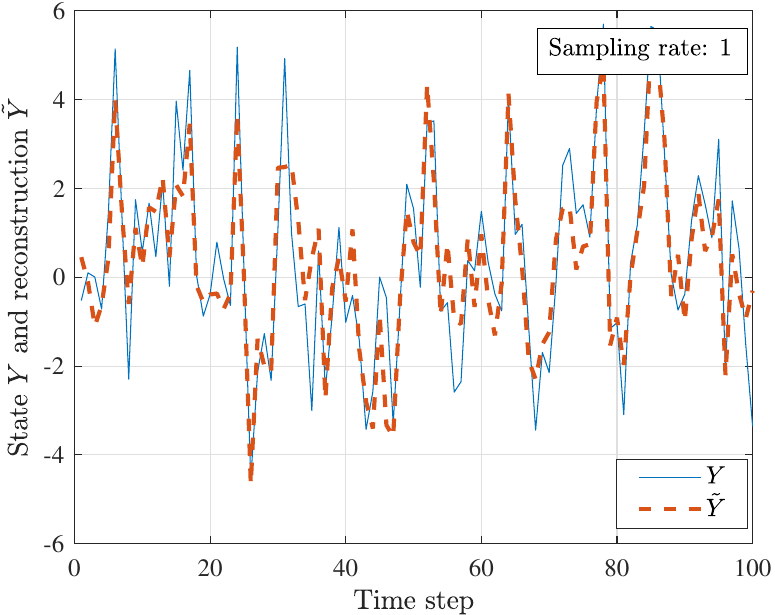}\label{Fig.EstOriY}
		}	
		\captionsetup{singlelinecheck = false, justification=raggedright}
		\caption{A trajectory of X and its reconstruction $(a)$, a trajectory of Y and its reconstruction $(b)$. The sampling rate is equal to one.}\label{Fig.EstOri}
		\vspace{-1em}
	\end{figure*}
	
	With Lemma \ref{Lm.PG}, we can approximately compute the gradient of  \eqref{Eq.ETParamOpt} with respect to $\theta$, once the reconstruction policy is the best response to the sampling policy, \emph{i.e.}, $\nabla _{\varphi}\tilde{L}\left( \theta ,\varphi^\star\left(\theta\right) \right) =0$. 
	\begin{remark}
		For nonlinear stochastic systems, the computation of information loss $\log \frac{p_{\theta}\left( Y^k \middle| Z^k \right)}{p_{\theta}\left( Y^k \middle| Z^{k-1} \right)}$ might be costly. We can apply variational techniques to develop information loss estimators to reduce complexities, \emph{e.g.}, \cite{weng2023optimal,tsur2023neural}.
	\end{remark}
	We develop the nested numerical algorithm \ref{Alg.ETESTOPT} to update the sampling and reconstruction policies based on their gradients \eqref{Eq.GradEst} and \eqref{Eq.ETPG}. Given the reconstruction policy $\tilde{\pi}_{\varphi}$, we first update the sampling policy $\pi_{\theta}$ with stochastic gradient methods that approximately compute the gradient \eqref{Eq.ETPG} with sample trajectories. Then, given $\theta$, we update the reconstruction policy based on the approximate version of \eqref{Eq.GradEst} until $\nabla _{\varphi}\tilde{L}\left( \theta ,\varphi \right) \approx 0$, which approximately ensures that the reconstruction is the best response to the sampler. We finally repeat this procedure until it converges.
	\begin{algorithm}
		\caption{Stochastic Policy Gradient Algorithm for Computing the Sampling and Reconstruction Policies} \label{Alg.ETESTOPT}
		\begin{algorithmic} 
			\State Initialize the sampling policy $\pi_\theta$, the reconstruction policy $\tilde{\pi}_\varphi$ and the positive constant $\alpha$, $\beta$.
			\Repeat
			\State {\bf Sampling policy update step:} 
			\State \quad Given $\theta$ and $\varphi$, generate sample trajectories.
			\State \quad Approximately update $\pi_\theta$ based on \eqref{Eq.ETPG}.
			\State \quad  {\bf Reconstruction policy update step:} 
			\RepeatUntil{}
			\State \qquad Given $\theta$  and $\varphi$, generate sample trajectories.
			\State \qquad approximately update $\tilde{\pi}_\varphi$ based on \eqref{Eq.GradEst}.
			\EndRepeat
			\Until{Convergence}
		\end{algorithmic}
	\end{algorithm}
	\section{Numerical Results}\label{Sec:Sim}
	In this section, we numerically study the performance of the developed sampling approach to privacy. To this end, consider the linear system \eqref{Eq.LinGauSys}, where matrix $A$ is given by 
	$$A=\left[ \begin{matrix}
		0.98&		-0.90\\
		0.00&		0.35\\
	\end{matrix} \right] ,$$
	and the covariance matrix for \emph{i.i.d.} Gaussian noises is
	$$Q=\left[ \begin{matrix}
		1.00&		0.10\\
		0.10&		4.00\\
	\end{matrix} \right]. $$
	The initial state is Gaussian distributed with zero mean and the covariance matrix
	$$Q_0=\left[ \begin{matrix}
		0.50&		0.25\\
		0.25&		0.50\\
	\end{matrix} \right]. $$
	In Fig. \ref{Fig.EstOri}, we illustrate the trajectories of \( X \) and \( Y \) over a horizon of \( K = 100 \), along with their reconstructed trajectories without downsampling. In Fig. \ref{Fig.EstOriX}, the state trajectory \( \{X_k\}_{k=0}^{100} \) is fully shared with the adversary. Given access to \( \{X_k\}_{k=0}^{100} \), the adversary can accurately estimate the private trajectory \( \{Y_k\}_{k=0}^{100} \), as shown in Fig. \ref{Fig.EstOriY}.  To mitigate privacy leakage and reduce storage requirements for state trajectories, we employ stochastic sampling techniques to selectively share or discard state values.  
	
	In the following simulations, we use the mean-squared error as the reconstruction metric for both \( X \) and \( Y \). The sampling rate, defined within the range \([0,1]\), represents the average number of sampled states over the horizon and serves as a measure of compression efficiency.  To achieve efficient compression and reduce storage requirements, the sampling rate should be kept low. If the sampling rate is zero, all state values are discarded, and no information is transmitted to the adversary. Conversely, a sampling rate of one implies that the full trajectory of \( X \) is shared, meaning no compression is applied.  
	
	We first evaluate the performance of stochastic samplers in two different settings: $\left(\romannumeral1\right)$ the optimal privacy-aware sampler and $\left(\romannumeral2\right)$ the open-loop sampler. In the second setting, the decision policy of the stochastic open-loop sampler from \cite{han2015stochastic} is similar to the privacy-aware design in \eqref{Eq.ETExp}, except that $g$ is set to zero and $f$ remains constant throughout the entire horizon. In other words, while the open-loop sampler is stochastic, the distribution of its sampling policy does not depend on the sampler's output. Based on the policy formulation in \eqref{Eq.ETExp}, the open-loop sampler is more likely to share $X_k$ with the adversary if $f$ is small or if the amplitude of $X_k$ is large. We adjust the value of the constant $f$ to explore the behavior of different open-loop samplers. As for the privacy-aware sampler, we optimize both $g$ and $f$ using the policy gradient algorithm.
	\begin{figure}[H]
		\centering
		\includegraphics[width=2.6in]{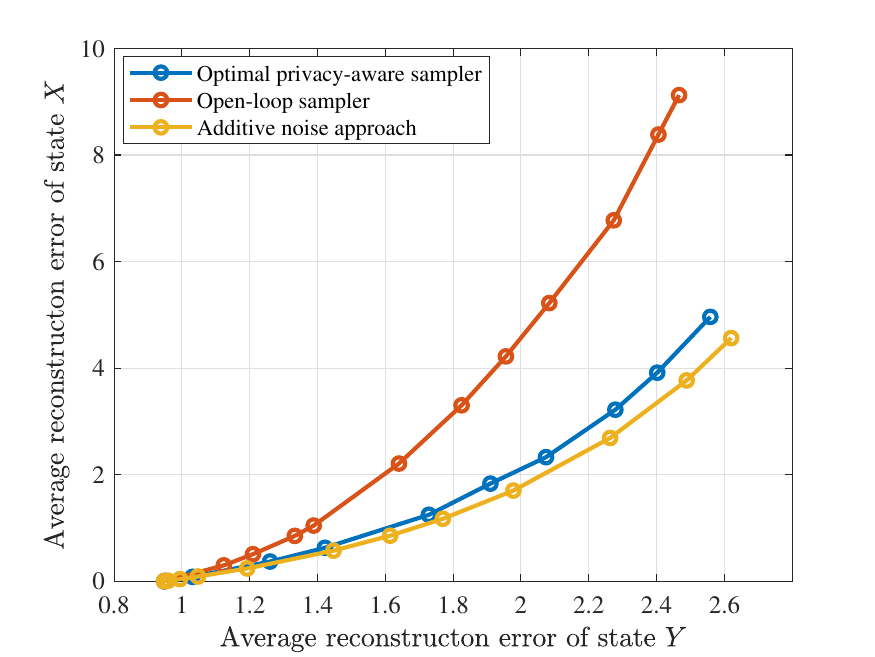}
		\caption{The average reconstruction error of $Y$ versus the average reconstruction error of $X$.}
		\label{Fig.Privacy_Distortion}	
	\end{figure}
	
	To compare the trade-off capabilities of different samplers, we plot the average reconstruction error of $Y$ versus the reconstruction error of $X$ over the entire horizon in Fig. \ref{Fig.Privacy_Distortion}. As shown in Fig. \ref{Fig.Privacy_Distortion}, the reconstruction error of $X$ increases with the adversary's uncertainty about $Y$, \emph{i.e.}, the reconstruction error of $Y$. This trade-off arises because our privacy protection approach is based on reducing the number of shared states. Furthermore, compared to the open-loop design, the optimal privacy-aware sampler achieves a smaller reconstruction error of $X$, demonstrating its superior capability in balancing utility and privacy protection.
	
	We also compare the performance of the privacy-aware sampler with the additive noise approach. In the additive noise approach, the state $X_k$ is perturbed by the Gaussian noise with zero mean before being transmitted, and then the state is estimated at the remote center using a Kalman filter. As the variance of the Gaussian noise increases, the amount of useful information that can be extracted from the perturbed state decreases, leading to a reduction in both utility and privacy leakage. However, the additive noise approach requires the sensor to communicate with the remote center at every time step, significantly increasing both the communication cost and the data storage volume. As illustrated in Fig. \ref{Fig.Privacy_Distortion}, the privacy-aware sampler demonstrates a competitive capability in achieving a trade-off between utility and privacy protection when compared to the additive noise approach. At the same time, the privacy-aware sampler generates fewer samples, thereby reducing communication overhead and storage requirements.
	\begin{figure*}
		\center
		\centering
		\subfigure[]{
			\centering
			\includegraphics[width=2.2in]{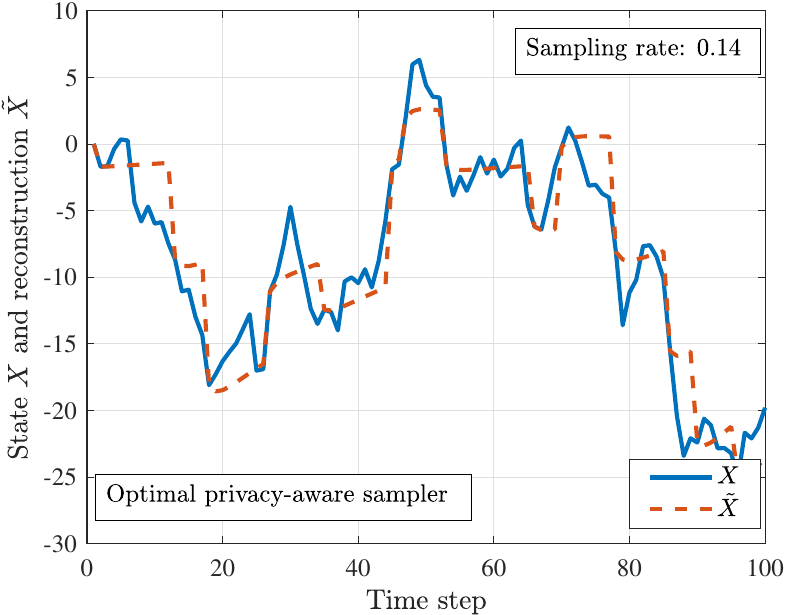}\label{Fig.PrivEstResX}	
		}
		\subfigure[]{
			\centering
			\includegraphics[width=2.2in]{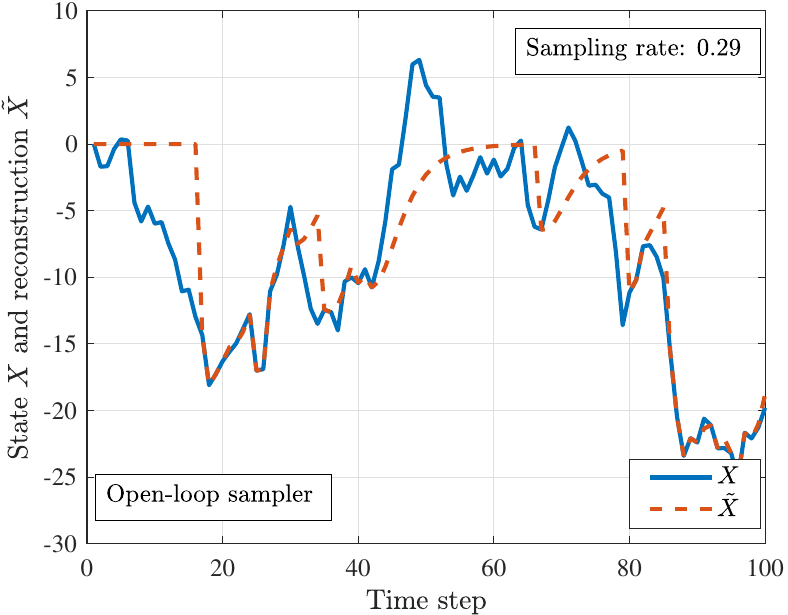}\label{Fig.UniEstResX}
		}	
		\subfigure[]{
			\centering
			\includegraphics[width=2.2in]{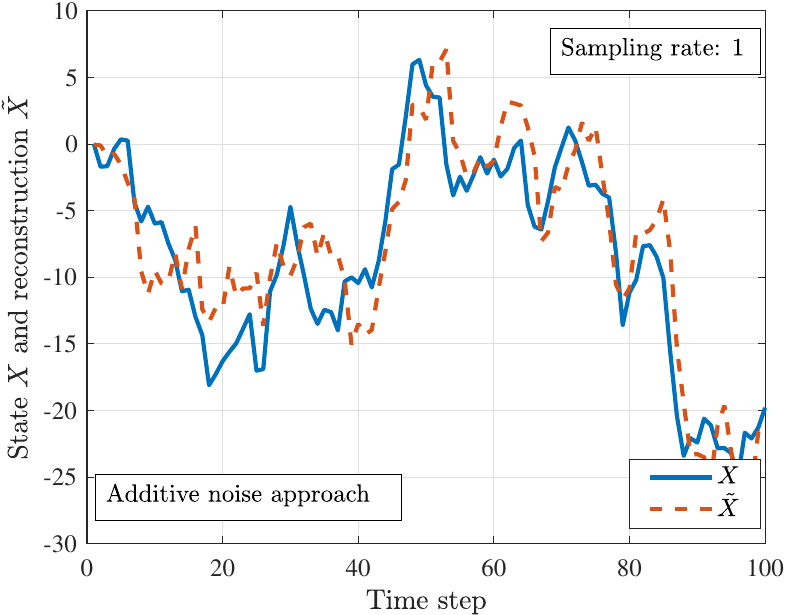}\label{Fig.AddNoiseEstResX}
		}
		
		\centering
		\subfigure[]{
			\centering
			\includegraphics[width=2.2in]{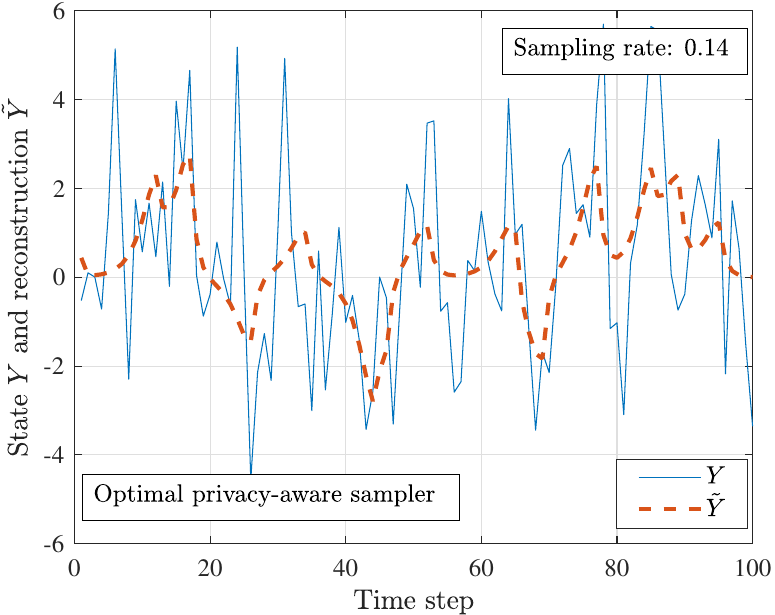}	\label{Fig.PrivEstResY}			
		}
		\subfigure[]{
			\centering
			\includegraphics[width=2.2in]{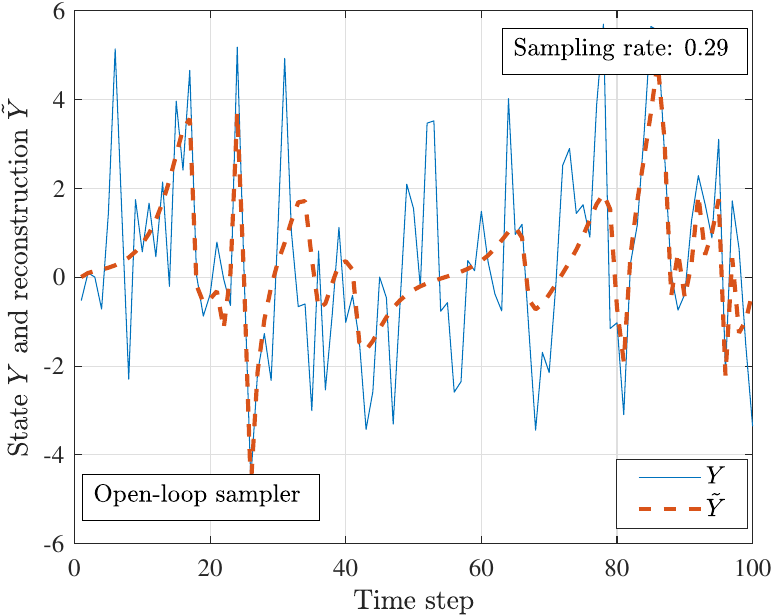}\label{Fig.UniEstResY}
		}
		\subfigure[]{
			\centering
			\includegraphics[width=2.2in]{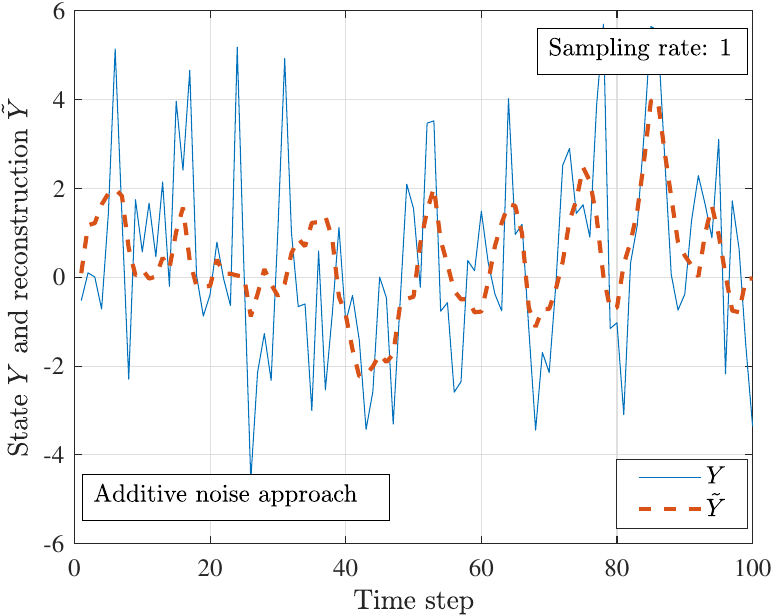}\label{Fig.AddNoiseEstResY}
		}
		\captionsetup{singlelinecheck = false, justification=raggedright}
		\caption{The reconstructed trajectories of $X$ ($a$) and $Y$ ($d$) under the optimal privacy-aware sampler. The reconstructed trajectories of $X$ ($b$) and $Y$ ($e$) under the open-loop sampler. The reconstructed trajectories of $X$ ($c$) and $Y$ ($f$) with the additive noise approach.} 
		\label{Fig.EstCmp}
		\vspace{-1em}
	\end{figure*}
	\begin{figure}[H]
		\centering
		\includegraphics[width=2.6in]{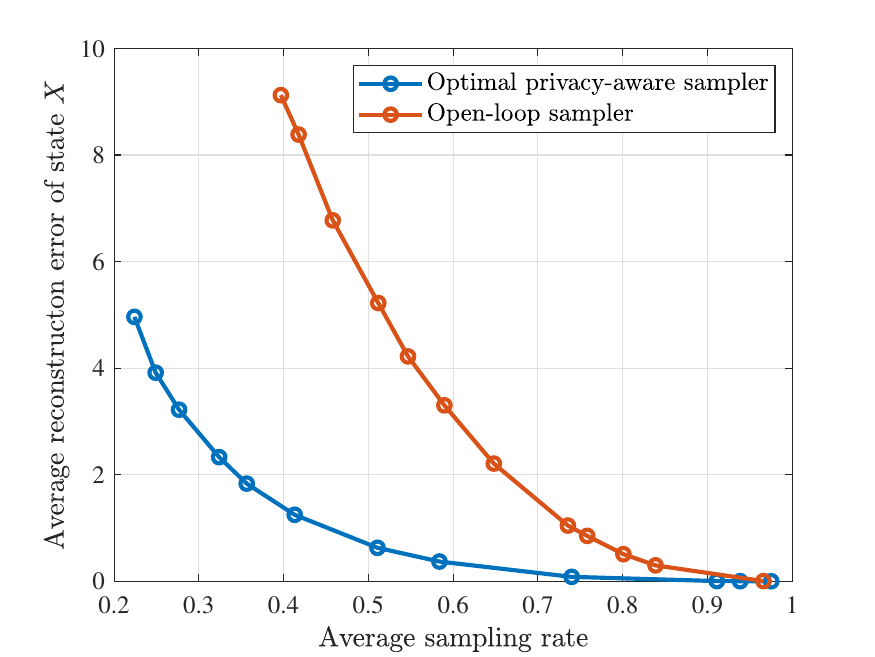}
		\caption{The average reconstruction error of $X$ under different sampling rates. }
		\label{Fig.Penalty_ETRate}	
	\end{figure}
	
	Besides, in Fig. \ref{Fig.EstCmp}, we present the optimal reconstruction trajectories based on the outputs of different samplers and the additive noise approach. The average sampling rates of the optimal privacy-aware sampler and the open-loop sampler are 0.14 and 0.29, respectively.  By comparing Fig. \ref{Fig.PrivEstResX} and Fig. \ref{Fig.UniEstResX}, we observe that the privacy-aware sampler outperforms the open-loop design in state reconstruction, even though the open-loop sampler uses more samples. Furthermore, as shown in Fig. \ref{Fig.PrivEstResY} and Fig. \ref{Fig.UniEstResY}, the adversary can recover the private trajectory more accurately under the open-loop design.  Additionally, by comparing Fig. \ref{Fig.PrivEstResX} with Fig. \ref{Fig.AddNoiseEstResX}, and Fig. \ref{Fig.PrivEstResY} with Fig. \ref{Fig.AddNoiseEstResY}, we observe that the optimal privacy-aware sampler achieves similar performance to the additive noise approach in reconstructing the state $X$ and protecting the privacy $Y$. 
	
	We compare the compression efficiency of different samplers in Fig. \ref{Fig.Penalty_ETRate}, which illustrates the average sampling rate versus the average reconstruction error of $X$. As shown in Fig. \ref{Fig.Penalty_ETRate}, the reconstruction error of $X$ increases as the sampling rate decreases, since state reconstruction becomes less efficient with fewer samples. Moreover, compared to the open-loop sampler, the optimal privacy-aware design achieves the same level of reconstruction accuracy with a smaller number of samples. This indicates that the privacy-aware design is more effective in reducing the state storage size while maintaining reconstruction quality.
	\section{Conclusions}\label{Sec:Con}
	In this paper, we developed a stochastic sampling framework for privacy protection in networked control and data compression. We presented the structural properties of the optimal privacy-aware sampler for both the nonlinear and linear systems via dynamic programming decomposition. We developed algorithms and verified our approach via simulations. The simulation results showed the compression efficiency of our approach, as well as its privacy and utility trade-off. 
	
	\bibliographystyle{IEEEtran}
	\bibliography{reference}
	\begin{IEEEbiography}[{\includegraphics[width=1in,height=1.2in,clip,keepaspectratio]{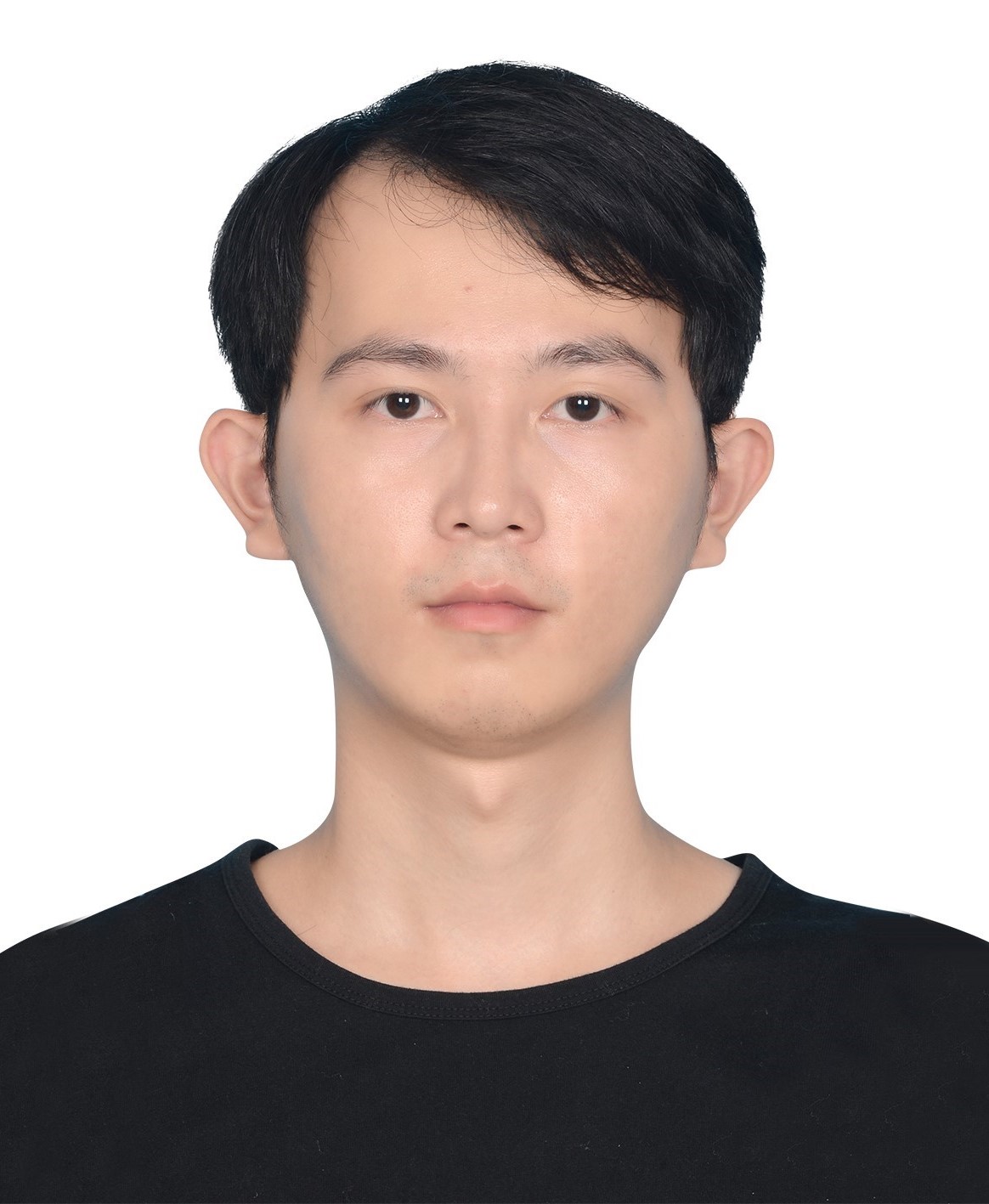}}]{Chuanghong Weng} is currently pursuing the
		Ph.D. degree with the Department of Electrical
		Engineering, City University of Hong Kong, Hong
		Kong. He received the M.S. degree from
		South China University of Technology, Guangzhou, China.
		
		His research interests include privacy in networked control systems and safe decision-makings.
	\end{IEEEbiography}
	\begin{IEEEbiography}[{\includegraphics[width=1in,height=1.2in,clip,keepaspectratio]{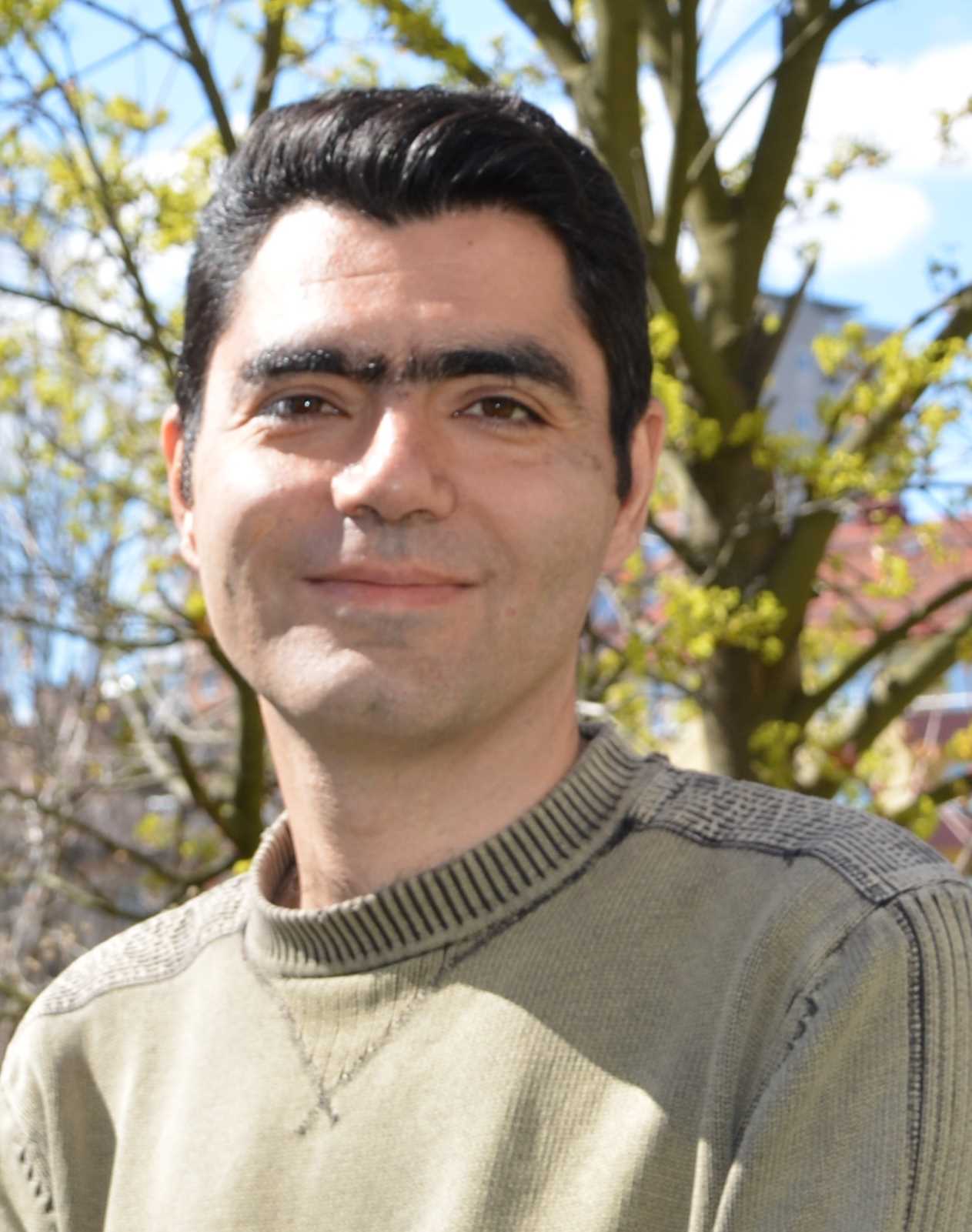}}]{Ehsan Nekouei}  received the
		B.S. degree from the Shahid Bahonar University
		of Kerman, Kerman, Iran, the M.S. degree from
		Tarbiat Modares University, Tehran, Iran, and
		the Ph.D. degree in electrical engineering from
		the University of Melbourne, Melbourne, VIC,
		Australia, in 2003, 2006, and 2013, respectively.
		
		He is currently an Assistant Professor with the
		Department of Electrical Engineering, City University of Hong Kong, Hong Kong. From 2014 to
		2019, he held postdoctoral positions at the KTH
		Royal Institute of Technology, Stockholm, Sweden, and the University of
		Melbourne, Melbourne, Australia.
		His research interests include privacy in networked control systems
		and integrated processing of human decision-making data.
	\end{IEEEbiography}
	\appendices	
	\section{Proof of Lemma \ref{Lm.OPEst}}\label{App.Lm.OPEst}
	Given the sampling policy, the mutual information is fixed. Therefore, optimizing the reconstructor reduces to minimizing the reconstruction loss.
	
	Since the reconstruction $\tilde{X}_k$ will not affect the evolution of $X_k$, the reconstruction history $\tilde{X}^{k-1}$ is redundant for the reconstruction optimization. Therefore, the reconstructor can greedily optimize the reconstruction loss at every time step. If $Z_k = \emptyset$, then the greedy optimization in \eqref{Eq.OPTEst} is optimal. If $Z_k = X_k$, since $l_D$ is non-negative, the optimal reconstruction $\tilde{X}_k^\star$ is equal to $Z_k$, which also be represented using \eqref{Eq.OPTEst}.
	
	According to \cite{Krishnamurthy_2016}, when the loss function $l_D\left( X_k, \tilde{X}_k \right)$ is defined as   $l_D\left( X_k, \tilde{X}_k \right) = \left( X_k - \tilde{X}_k \right)^\top \left( X_k - \tilde{X}_k \right)$, 
		the optimal solution is given by the conditional expectation 
		$\mathsf{E}\left[ \left. X_k \right| Z^{k} \right]$. 
	
	\begin{figure*}
		\begin{align}\label{Eq.AppVk1}
			\setcounter{equation}{27}
			\mathsf{E}\left[ V_{k+1}^{\star}\left( b_{k+1} \right) \middle| Z^{k-1} \right] &=P\left( N_k=0 \middle| Z^{k-1} \right) V_{k+1}^{\star}\left( \varPhi \left( b_k,\mathcal{A} _k, Z_k=\emptyset \right) \right) \nonumber
			\\
			&\quad+\int{P\left( N_k=1 \middle| x_k,Z^{k-1} \right) p\left( x_k \middle| Z^{k-1} \right) V_{k+1}^{\star}\left( \varPhi \left( b_k,\mathcal{A} _k,Z_k=x_k \right) \right) dx_k} \nonumber
			\\
			&=\int{\int{\int{a_k\left( N_k=0 \middle| x_k,m_k \right) b_k\left( x_k,m_k,y^k|Z^{k-1} \right)}V_{k+1}^{\star}\left( \varPhi \left( b_k,\mathcal{A} _k,Z_k=\emptyset \right) \right) dx_kdm_kdy^k}} \nonumber
			\\
			&+\int{\!\!\int{\!\!\int{a_k\left( N_k=1 \middle| x_k,m_k \right) b_k\left( x_k,m_k,y^k|Z^{k-1} \right) V_{k+1}^{\star}\left( \varPhi \left( b_k,\mathcal{A} _k,Z_k=x_k \right) \right) dx_kdm_kdy^k}}}.
		\end{align}
		\hrule
	\end{figure*}
	\section{Proof of Lemma \ref{Lm.Equivalent}}\label{App.Lm.Equivalent}
	Given the sampling policy $\pi = \{\pi_k\}_{k=0}^K$, we construct the policy collection $\mathcal{A} = \{\mathcal{A}_k\left(Z^{k-1}\right)\}_{k=0}^K$ as follows, 
		\small
		\begin{align}\label{Eq.conPiFromA}
			\mathcal{A} _k\!\left(\! Z^{k-1} \!\right) \!=\!\left\{\! \begin{array}{c}
				\!a_k\!\left( n_k \middle| x_k,m_k \right) \!=\!\pi _k\left( n_k \middle| x_k,m_k,Z^{k-1} \right) ,\\
				\forall n_k,x_k,m_k\!\\
			\end{array} \!\right\} \!, \!
		\end{align}  
		\normalsize
		which clusters policies that share the same information $Z^{k-1}$. Consequently, the auxiliary optimization objective \eqref{Eq.OP2} has the same value as the original optimization objective \eqref{Eq.OP}, \emph{i.e.}, $L\left(\mathcal{A}\right) = L\left(\pi\right)$.
	
	Conversely, given the policy collection $\mathcal{A} = \{\mathcal{A}_k\left(Z^{k-1}\right)\}_{k=0}^K$, which is implicitly dependent on $Z^{k-1}$, we define the sampling policy as  
		\begin{align}
			\setcounter{equation}{26}
			\pi_{k}\left(\left. {N_k}\right|{X_k, M_k, Z^{k-1}}\right) = a_k\left(\left. N_k\right|X_k, M_k\right), \nonumber
		\end{align}  
		and for $k=0$,  
		\begin{align}
			\pi_{0}\left(\left. {N_0}\right|{X_0}\right) = a_0\left(\left. N_0\right|X_0\right). \nonumber
		\end{align}  
		Clearly, the original optimization objective \eqref{Eq.OP} and the auxiliary optimization objective \eqref{Eq.OP2} yield the same value, \emph{i.e.}, $L\left(\pi\right) = L\left(\mathcal{A}\right)$.  
	
	This equivalence demonstrates that the optimal policy collection \( \mathcal{A}^\star \) minimizing \eqref{Eq.OP2} can be used to derive the optimal sampling policy \( \pi^\star \) minimizing \eqref{Eq.OP}. We next prove this by contradiction.  
	
	Assume that the minimum value of the objective function \eqref{Eq.OP2} is \( L^\star_1\left(\mathcal{A}_1^\star\right) \), with the corresponding optimal policy collection \( \mathcal{A}_1^\star \). If the policy \( \pi_1^\star \) constructed from \( \mathcal{A}_1^\star \) is not optimal for \eqref{Eq.OP}, then there exists another policy \( \pi_2 \) such that \( L_2\left(\pi_2\right) < L^\star_1\left(\pi_1^\star\right) \).  
	
	According to \eqref{Eq.conPiFromA}, we can construct a policy collection \( \mathcal{A}_2 \) such that \( L_2\left(\mathcal{A}_2\right) = L_2\left(\pi_2\right) \). Given that \( L^\star_1\left(\mathcal{A}_1^\star\right) = L^\star_1\left(\pi_1^\star\right) \), \( L_2\left(\mathcal{A}_2\right) = L_2\left(\pi_2\right) \) and \( L^\star_2\left(\pi_2\right) < L^\star_1\left(\pi_1^\star\right) \), it follows that \( L_2\left(\mathcal{A}_2\right) < L^\star_1\left(\mathcal{A}_1^\star\right) \), contradicting the assumption that \( \mathcal{A}_1^\star \) is the optimal policy collection.  
	\section{Proof of Theorem \ref{Th.OPTEQU}}\label{App.Th.OPTEQU}
	To prove Theorem \ref{Th.OPTEQU}, we first rewrite the auxiliary optimization problem by expanding the mutual information into multiple terms along time.
		\begin{lemma} \label{Lm.OptProblem}
			The privacy-aware sampling policy design can be written as
			\small
			\begin{align}
				\!\!\min_{\mathcal{A}}\! \mathsf{E}\!\!\left[\! \sum_{k=0}^K\!{\mathsf{E}\!\left[ l_D\!\left(\! X_k,\tilde{\pi}_{k}^{\star}\!\left( Z^k \right) \!\right) \!\middle| Z^{k-1} \right] \!\!+\!\lambda \mathsf{E}\!\!\left[ \! \log\! \frac{p\left( Y^k \middle| Z^k \right)}{p\!\left(\! Y^k \middle| Z^{k-1} \right)} \middle| \! Z^{k-1} \! \right]} \!\right]\!. \nonumber
			\end{align}
			\normalsize
		\end{lemma}
	\begin{proof}
		See in Appendix \ref{App.Lm.OptProblem}.
	\end{proof}
	We next show the one-step loss in Lemma \ref{Lm.OptProblem} is a function of the policy collection $\mathcal{A}_k$ and the belief state $b_k$.
	\begin{lemma}\label{Lm.beliefLoss}
		The one-step loss can be expressed as
		\begin{align}
			&\mathsf{E}\!\left[ l_D\left( X_k,\tilde{\pi}_{k}^{\star}\left( Z^k \right) \right) \middle| Z^{k-1} \right] +\lambda \mathsf{E}\!\left[ \log \frac{p\left( Y^k \middle| Z^k \right)}{p\left( Y^k \middle| Z^{k-1} \right)} \middle| Z^{k-1}\! \right]\nonumber \\ &= \tilde{l}_D\left( \mathcal{A} _k,b_k,\tilde{\pi}_{k}^{\star}\left( Z_k=\emptyset,Z^{k-1} \right) \right)+\lambda l_I\left( \mathcal{A} _k,b_k \right), \nonumber
		\end{align}
		where $l_D$ and $l_I$ are functions of $\mathcal{A}_k$ and $b_k$, and the belief state $b_{k+1}$ can be updated via \eqref{Eq.BSUP} . 
	\end{lemma}
	\begin{proof}
		See in Appendix \ref{App.Lm.beliefLoss}.
	\end{proof}
	We next prove the optimality equation in Theorem \ref{Th.OPTEQU} based on Lemma \ref{Lm.OptProblem} and \ref{Lm.beliefLoss}.
	
	At time $K$, with the shared information $Z^{K-1}$, the optimal cost-to-go function can be defined as \eqref{Eq.ValueFuncK}.
		\small
		\begin{align}\label{Eq.ValueFuncK}
			&V_{K}^{\star}\left( Z^{K-1} \right) =\min_{\mathcal{A} _K} \mathsf{E}\left[ l_D\left( X_K,\tilde{\pi}_{K}^{\star}\left( Z^K \right) \right) \middle| Z^{K-1} \right]\nonumber\\ &\!\!\!\!\!+\!\lambda \mathsf{E}\!\left[ \!\log \frac{p\left( Y^K \middle| Z^K \right)}{p\left( Y^K \middle| Z^{K-1} \right)} \middle| Z^{K-1} \right] \!+\! \mathsf{E}\left[\!\left.V_{K+1}^\star\!\left(Z^K\right)\right| Z^{K-1}\right]\!, 
		\end{align}
		\normalsize
		with $V_{K+1}^\star\left(Z^K\right)\overset{\bigtriangleup}{=}0$. 
		According to Lemma \ref{Lm.beliefLoss}, $V_{K}^{\star}\left( Z^{K-1} \right)$ is determined by $b_K$, thus can be expressed as $V_{K}^{\star}\left( b_K \right)$.
	
	We next prove that $V_{k}^{\star}\left( Z^{k-1} \right)$ is only determined by $b_k$ if $V_{k+1}^{\star}\left( Z^{k} \right)$ only depends on $b_{k+1}$, by induction. According to the optimality principle, we can express the optimal cost-to-go function $V_{k}^{\star}\left( Z^{k-1} \right)$ with
	\begin{align}
		&V_{k}^{\star}\left( Z^{k-1} \right) =\min_{\mathcal{A} _k} \mathsf{E}\left[ l_D\left( X_k,\tilde{\pi}_{k}^{\star}\left( Z^k \right) \right) \middle| Z^{k-1} \right] \nonumber\\+&\lambda \mathsf{E}\!\left[\! \log \frac{p\left( Y^k \middle| Z^k \right)}{p\left( Y^k \middle| Z^{k-1} \right)} \middle| Z^{k-1} \!\right] \!+\!\mathsf{E}\!\left[\! V_{k+1}^{\star}\left( b_{k+1} \right) \middle| Z^{k-1} \!\right] . \nonumber
	\end{align}
	where the one-step loss depends on $b_k$ as shown in Lemma \ref{Lm.beliefLoss}.  $\mathsf{E}\left[ V_{k+1}^{\star}\left( b_{k+1} \right) \middle| Z^{k-1} \right]$ also depends on $b_k$ as shown in \eqref{Eq.AppVk1}.
	Therefore, the optimal cost-to-go function can be denoted with $V_{k}^{\star}\left( b_k \right)$.
	
	\section{Proof of Lemma \ref{Lm.GaussianBelief}}\label{App.Lm.GaussianBelief}
	\begin{figure*}
		\begin{align}\label{Eq.AppJointGau}
			\setcounter{equation}{29}
			&\pi _k\left( \left. N_k=0 \right|X_k,Z^{k-1} \right) \nonumber \\ 
			=&\exp \left( -\frac{1}{2}\left( X_k-g_k\left( Z^{k-1} \right) \right) ^{\top}f_{k}^{-1}\left( Z^{k-1} \right) \left( X_k-g_k\left( Z^{k-1} \right) \right) \right)  
			\nonumber \\
			=&\exp\! \left[\! -\frac{1}{2}\left(\! \left[ \begin{array}{c}
				\!X_k\!\\
				\!Y^k\!\\
			\end{array} \right] \!-\!\left[ \begin{array}{c}
				\!g_k\left(Z^{k-1}\right)\!\\
				\!0_{n_y\left( k+1 \right) \times 1}\!\\
			\end{array} \right] \right) ^{\!\!\top}\!\!\left[ \begin{matrix}
				f_{k}^{-1}\left( Z^{k-1} \right)&		0_{n_x\times n_y\left( k+1 \right)}\\
				0_{n_y\left( k+1 \right) \times n_x}&		0_{n_y\left( k+1 \right) \times n_y\left( k+1 \right)}\\
			\end{matrix} \right]\!\! \left( \left[ \begin{array}{c}
				\!X_k\!\\
				\!Y^k\!\\
			\end{array} \right] \!-\!\left[ \begin{array}{c}
				\!g_k\left(Z^{k-1}\right)\!\\
				\!0_{n_y\left( k+1 \right) \times 1}\!\\
			\end{array} \right] \right) \!\right] .
		\end{align}
		\hrule
	\end{figure*}
	\begin{figure*}
		\begin{align} \label{Eq.GauNkis0}
			p\left( \left. X_{k+1},Y^{k+1} \right|Z^{k} \right) &=\int{p\left( X_{k+1},Y^{k+1} \middle| x_k,Y^{k},Z^{k} \right) p\left( x_k,Y^{k} \middle| Z^{k} \right) dx_k} \nonumber
			\\
			&\overset{(a)}{=}\int{p\left( X_{k+1},Y^{k+1} \middle| x_k,Y^{k} \right) p\left( x_k,Y^{k} \middle| N_k=0,Z^{k-1} \right) dx_k} \nonumber
			\\
			&=\int{\mathcal{N} \left( \left[ \begin{array}{c}
					X_{k+1}\\
					Y^{k+1}\\
				\end{array} \right] ;A_{k+1}\left[ \begin{array}{c}
					x_k\\
					Y^{k}\\
				\end{array} \right] ,\varSigma _{k+1} \right) \mathcal{N} \left( \left[ \begin{array}{c}
					x_k\\
					Y^{k}\\
				\end{array} \right] ;\left[ \begin{array}{c}
					\tilde{x}_{k|k}\\
					\tilde{y}^{k|k}\\
				\end{array} \right] ,P_{k|k} \right) dx_k},
		\end{align} 
		\hrule
		\begin{align}\label{Eq.AppProYkP1}
			\setcounter{equation}{35}
			p\left( \left. X_{k+1},Y^{k+1} \right|Z^k \right) &=\int{p\left( X_{k+1},Y^{k+1} \middle| x_k,Y^k,Z^k \right) p\left( x_k,Y^k \middle| Z^k \right) dx_k}\nonumber
			\\
			&=\int{p\left( X_{k+1},Y^{k+1} \middle| x_k,Y^k \right) p\left( x_k,Y^k \middle| Z^k \right) dx_k}\nonumber
			\\
			&=\int{\mathcal{N} \left( \left[ \begin{array}{c}
					X_{k+1}\\
					Y^{k+1}\\
				\end{array} \right] ;A_{k+1}\left[ \begin{array}{c}
					x_k\\
					Y^k\\
				\end{array} \right] ,\varSigma _{k+1} \right) \mathcal{N} \left( \left[ \begin{array}{c}
					x_k\\
					Y^k\\
				\end{array} \right] ;\left[ \begin{array}{c}
					Z_k\\
					\tilde{y}^k\\
				\end{array} \right] ,P_{k|k} \right)dx_k}.
		\end{align}
		\hrule
	\end{figure*}
			
	We next prove Lemma \ref{Lm.GaussianBelief} by induction. Specially, we prove that $p\left( \left. X_k,Y^k \right|Z^{k} \right)$ and $p\left( \left. X_{k+1},Y^{k+1} \right|Z^k \right) $ are Gaussian with the assumption that $p\left( \left. X_k,Y^k \right|Z^{k-1} \right)$ is Gaussian. 
			
	For $k=0$, we have that $p\left(X_0,Y_0\right)$ is a joint Gaussian density distribution, which is one of assumptions about the system model.
			
	We assume $p\left( \left. X_k,Y^k \right|Z^{k-1} \right)$ is Gaussian defined as follows,
				\begin{align}
					p\left( \left. X_k,Y^k \right|Z^{k-1} \right) =\mathcal{N} \left( \left[ \begin{array}{c}
						X_k\\
						Y^k\\
					\end{array} \right] ;\left[ \begin{array}{c}
						\tilde{x}_{k|k-1}\\
						\tilde{y}^{k|k-1}\\
					\end{array} \right] ,P_{k|k-1} \right) , \nonumber
				\end{align}	
				and then prove that $p\left( \left. X_k,Y^k \right|Z^k \right) $ and $p\left( \left. X_{k+1},Y^{k+1} \right|Z^{k+1} \right) $ are also Gaussian. 
			
			If $N_k=0$, then
				\begin{align} \label{Eq.AppUpNk0}
					\setcounter{equation}{28}
					&p\left( \left. X_k,Y^k \right|Z^k \right) \nonumber\\=&p\left( \left. X_k,Y^k \right|N_k=0,Z^{k-1} \right) \nonumber
					\\
					=&\frac{\pi _k\left( \left. N_k=0 \right|X_k,Z^{k-1} \right) p\left( \left. X_k,Y^k \right|Z^{k-1} \right)}{P\left( \left. N_k=0 \right|Z^{k-1} \right)}, 
				\end{align}
				where $\pi _k\left( \left. N_k=0 \right|X_k,Z^{k-1} \right) $ can be expressed as a density function similar to the unnormalized joint Guassian density of $X_k$ and $Y^k$ in \eqref{Eq.AppJointGau}.	As stated in \cite{petersen2008matrix}, the normalized product of two Gaussian density functions is also Gaussian and can be derived analytically using the rules provided in \cite{petersen2008matrix}. Also, with $\int \int p\left( X_k, Y^k \mid Z^k \right) \, dX_k dY^k = 1$, consequently, $p\left( \left. X_k,Y^k \right|Z^k \right)$ in \eqref{Eq.AppUpNk0} is a Gaussian density function from the normalized product of $\pi _k\left( \left. N_k=0 \right|X_k,Z^{k-1} \right) $ and $p\left( \left. X_k,Y^k \right|Z^{k-1} \right)$. As a result, $p\left( \left. X_k,Y^k \right|Z^k \right)$ can be characterized by the following conditional mean and covariance
				\begin{align}
					\!\!\left[\!\! \begin{array}{c}
						\tilde{x}_{k|k}\\
						\tilde{y}^{k|k}\\
					\end{array} \!\!\right] \!=\!\left( D_k+I \right) ^{-1}&\!\left(\! D_k\!\left[\!\! \begin{array}{c}
						g_k\left( Z^{k-1} \right)\\
						0_{n_y\left( k+1 \right) \times 1}\\
					\end{array} \!\!\right] \!+\!\left[\!\! \begin{array}{c}
						\tilde{x}_{k|k-1}\\
						\tilde{y}^{k|k-1}\\
					\end{array} \!\!\right] \right) \!,\nonumber\\
					P_{k|k}=&\left( D_k+I \right) ^{-1}P_{k|k-1},\nonumber
				\end{align}
				where
				\begin{align}
					D_k=P_{k|k-1}\left[ \begin{matrix}
						f_{k}^{-1}\left( Z^{k-1} \right)&		0_{n_x\times n_y\left( k+1 \right)}\\
						0_{n_y\left( k+1 \right) \times n_x}&		0_{n_y\left( k+1 \right) \times n_y\left( k+1 \right)}\\
					\end{matrix} \right] . \nonumber
			\end{align}
			If $N_k=1$, then we have,
				\begin{align}
					&p\left( \left. X_k,Y^k \right|Z^k \right)\nonumber\\ =&p\left( \left. Y^k \right|N_k=1,X_k,Z^{k-1} \right) \delta \left( Z_k-X_k \right) \nonumber
					\\
					=&\frac{P\left( \left. N_k=1 \right|X_k,Y^k,Z^{k-1} \right) p\left( \left. Y^k \right|X_k,Z^{k-1} \right)}{\pi _k\left( \left. N_k=1 \right|X_k,Z^{k-1} \right)}\delta \left( Z_k-X_k \right) \nonumber
					\\
					=&p\left( \left. Y^k \right|X_k,Z^{k-1} \right) \delta \left( Z_k-X_k \right) ,\nonumber
				\end{align}
				where the last equality holds since the sampling policy $\pi_{k}$ is fully dependent on $X_{k}$ and $Z^{k-1}$, i.e., $P\left( \left. N_k=1 \right|X_k,Y^k,Z^{k-1} \right)=\pi _k\left( \left. N_k=1 \right|X_k,Z^{k-1} \right)$. Since $p\left( X_{k},Y^{k} \middle| Z^{k-1} \right)$ is a joint Gaussian density function, the conditional density $p\left( Y^{k} \middle| X_{k},Z^{k-1} \right)$ is a Gaussian density function of $Y^{k}$. According to the conditional density computation rule \cite{petersen2008matrix}, we have
				\begin{align}
					\tilde{y}^{k|k}=\tilde{y}^{k|k-1}+P_{k|k-1}^{yx}\left( P_{k|k-1}^{xx} \right) ^{-1}\left( X_k-\tilde{x}_{k|k-1} \right) .\nonumber
				\end{align}
				Also, given $Z_k\neq \emptyset$, $X_k$ is deterministic, which can be considered as a special Gaussian variable with the mean $Z_k$ and zero covariance. As a result, $p\left( \left. X_k,Y^k \right|Z^k \right)$ can be regarded as a joint Gaussian density function with $\tilde{x}_{k|k}=Z_{k}$ and
				\begin{align}
					P_{k|k}=\left[ \begin{matrix}
						0_{n_x\times n_x}&		0_{n_x\times n_y\left( k+1 \right)}\\
						0_{n_y\left( k+1 \right) \times n_x}&		P_{k|k}^{yy}\\
					\end{matrix} \right] ,\nonumber
				\end{align}
				where
				\begin{align}
					P_{k|k}^{yy}=P_{k|k-1}^{yy}-P_{k|k-1}^{yx}\left( P_{k|k-1}^{xx} \right) ^{-1}P_{k|k-1}^{xy}.\nonumber
			\end{align}
			\begin{figure*}
				\begin{align}\label{Eq.AppProbNk}
					\setcounter{equation}{36}
					&P\left( N_k=0 \middle| Z^{k-1} \right) \nonumber \\=&\int{P\left( N_k=0 \middle| x_k,Z^{k-1} \right) p\left( x_k \middle| Z^{k-1} \right) dx_k} \nonumber
					\\
					=&\sqrt{\left| 2\pi f_k\left( Z^{k-1} \right) \right|}\int{\mathcal{N} \left( x_k;g_k\left( Z^{k-1} \right) ,f_k\left( Z^{k-1} \right) \right) \mathcal{N} \left( x_k;\tilde{x}_{k|k-1},P_{k|k-1}^{xx} \right) dx_k} \nonumber
					\\
					=&\sqrt{\frac{\left| f_k\left( Z^{k-1} \right) \right|}{\left| \left( f_k\left( Z^{k-1} \right) +P_{k|k-1}^{xx} \right) \right|}}\exp \left( -\frac{1}{2}\left( g_k\left( Z^{k-1} \right) -\tilde{x}_{k|k-1} \right) ^{\top}\left( f_k\left( Z^{k-1} \right) +P_{k|k-1}^{xx} \right) ^{-1}\left( g_k\left( Z^{k-1} \right) -\tilde{x}_{k|k-1} \right) \right).
				\end{align}
				\hrule
				\begin{align}\label{Eq.AppInf}
					\mathsf{E}\left[ \frac{\log p\left( Y^k \middle| Z^k \right)}{\log p\left( Y^k \middle| Z^{k-1} \right)} \middle| Z^{k-1} \right] =&\int{\int{P\left( N_k=1 \middle| x_k,Z^{k-1} \right) p\left( x_k,y^k \middle| Z^{k-1} \right) \log p\left( y^k \middle| N_k=1,x_k,Z^{k-1} \right)}}dx_kdy^k \nonumber
					\\
					\quad+&P\left( N_k=0 \middle| Z^{k-1} \right) \int{p\left( y^k \middle| N_k=0,Z^{k-1} \right) \log p\left( y^k \middle| N_k=0,Z^{k-1} \right)}dy^k  \nonumber
					\\
					-&\int{p\left( y^k \middle| Z^{k-1} \right) \log p\left( y^k \middle| Z^{k-1} \right)}dy^k.
				\end{align}
				\hrule
			\end{figure*}
			
			We next prove that $p\left( \left. X_{k+1},Y^{k+1} \right|Z^k \right) $ is Gaussian. 
				If $N_k$ is 0, then we have \eqref{Eq.GauNkis0}, where $(a)$ follows the Markov chain $Z^k\to \left\{ X_k, Y^k\right\} \to \left\{ X_{k+1}, Y^{k+1}\right\}$. According to Theorem 3.4.2 in \cite{Krishnamurthy_2016}, the marginal result in \eqref{Eq.GauNkis0} is a joint Gaussian density function with the following mean and covariance,
				\begin{align} \label{Eq.GauKToKP1_1}
					\setcounter{equation}{31}
					\left[ \begin{array}{c}
						\tilde{x}_{k+1|k}\\
						\tilde{y}^{k+1|k}\\
					\end{array} \right] =A_{k+1}\left[ \begin{array}{c}
						\tilde{x}_{k|k}\\
						\tilde{y}^{k|k}\\
					\end{array} \right] ,
				\end{align}
				\begin{align} \label{Eq.GauKToKP1_2}
					P_{k+1|k}=A_{k+1}P_{k|k}A_{k+1}^{\top}+\varSigma _{k+1},
				\end{align}
				where
				\begin{align} \label{Eq.GauKToKP1_3}
					A_{k+1}=\left[ \begin{matrix}
						A&		0_{n\times kn_y}\\
						0_{n_y \left(k+1\right) \times n_x}&		I_{n_y\left( k+1 \right) \times n_y\left( k+1 \right)}\\
					\end{matrix} \right] ,
				\end{align}
				\begin{align}  \label{Eq.GauKToKP1_4}
					\!\!\!\!\!\!\!\!\!	\varSigma _{k+1}=\left[ \begin{matrix}
						Q&		0_{n\times n_y\left( k+1 \right)}\\
						0_{n_y\left( k+1 \right) \times n}&		0_{n_y\left( k+1 \right) \times n_y\left( k+1 \right)}\\
					\end{matrix} \right] .
				\end{align}
				When \( N_k = 1 \), we obtain \eqref{Eq.AppProYkP1}, which is analogous to the case where \( N_k = 0 \) in \eqref{Eq.GauNkis0}. Consequently, the same update rule for \( p\left( X_{k+1}, Y^{k+1} \middle| Z^k \right) \) follows as in \eqref{Eq.GauKToKP1_1}–\eqref{Eq.GauKToKP1_4}, except that \( \tilde{x}_{k|k} = Z_k \).
			\section{Proof of Lemma \ref{Lm.GaussianOpt}}\label{App.Lm.GaussianOpt}
			We first define the one-step loss as,
			\begin{align}
				l_k\left( \pi_k,Z^{k-1} \right) =&\mathsf{E}\left[ \left( X_k-\tilde{X}_k \right) ^{\top}\left( X_k-\tilde{X}_k \right) \middle| Z^{k-1} \right]\nonumber \\ +&\mathsf{E}\left[ \frac{\log p\left( Y^k \middle| Z^k \right)}{\log p\left( Y^k \middle| Z^{k-1} \right)} \middle| Z^{k-1} \right] .\nonumber
			\end{align}
			Obviously, the optimization problem is equivalent to 
			\begin{align}
				\min_{\left\{ \pi _k \right\} _{k=0}^{K}} L\left( \pi \right) =\min_{\left\{ \pi _k \right\} _{k=0}^{K}} \mathsf{E}\left[ \sum_{k=0}^K{l_k\left( \pi_k,Z^{k-1} \right)} \right] .\nonumber
			\end{align}
			
			We next show that $l_k\left( \pi_k,Z^{k-1} \right)$ can be analytically expressed with $\tilde{x}_{k|k-1}$ and $P_{k|k-1}$. First, for the conditional mean-squared error, we have,
			\begin{align}
				&\mathsf{E}\left[ \left( X_k-\tilde{X}_k \right) ^{\top}\left( X_k-\tilde{X}_k \right) \middle| Z^{k-1} \right] \nonumber\\=&P\!\left(\! N_k\!=\!0 \middle| Z^{k-1} \!\right) \!\mathsf{E}\!\left[\! \left(\! X_k-\tilde{X}_k \!\right) ^{\top}\!\!\left(\! X_k-\tilde{X}_k\! \right) \middle| N_k=0,Z^{k-1} \!\right] \nonumber
				\\
				=&P\left( N_k=0 \middle| Z^{k-1} \right) P_{k|k}^{xx}.\nonumber
			\end{align}
			According to the update rules in Lemma \ref{Lm.GaussianBelief}, if $N_k=0$, then $P_{k|k}$ can be written as
				\small
				\begin{align}
					& \left[\! \begin{matrix}
						P_{k|k}^{xx}&		\!\!P_{k|k}^{xy}\\
						P_{k|k}^{yx}&	\!\!	P_{k|k}^{yy}\\
					\end{matrix} \!\right] \nonumber \\
					=&\!\left(\! P_{k|k-1}\left[ \begin{matrix}
						f_{k}^{-1}\left( Z^{k-1} \right)&		0_{n_x\times n_y\left( k+1 \right)}\\
						0_{n_y\left( k+1 \right) \times n_x}&		0_{n_y\left( k+1 \right) \times n_y\left( k+1 \right)}\\
					\end{matrix} \right] +I \!\right) ^{-1}\!P_{k|k-1}\nonumber
					\\
					=&\!\left[\! \begin{matrix}
						P_{k|k-1}^{xx}f_{k}^{-1}\left( Z^{k-1} \right) +I&		0_{n_x\times n_y\left( k+1 \right)}\\
						P_{k|k-1}^{yx}f_{k}^{-1}\left( Z^{k-1} \right)&		I\\
					\end{matrix} \!\right] ^{-1}\!\left[\! \begin{matrix}
						P_{k|k-1}^{xx}&		\!\!P_{k|k-1}^{xy}\\
						P_{k|k-1}^{yx}&	\!\!	P_{k|k-1}^{yy}\\
					\end{matrix} \!\right], \nonumber.
				\end{align}
				\normalsize
				Based on the inverse rule of block matrices in \cite{petersen2008matrix}, the covariance update equations are given by
				\begin{align}
					P_{k|k}^{xx} &= f_k\left( Z^{k-1} \right) \left( f_k\left( Z^{k-1} \right) + P_{k|k-1}^{xx} \right)^{-1} P_{k|k-1}^{xx}, \nonumber
				\end{align}
				and
				\begin{align}
					P_{k|k}^{yy} &= P_{k|k-1}^{yy} - P_{k|k-1}^{yx} \left( f_k\left( Z^{k-1} \right) + P_{k|k-1}^{xx} \right)^{-1} P_{k|k-1}^{xy}. \nonumber
				\end{align}
				Thus, the conditional expectation of the squared estimation error is given by
				\begin{align}
					&\mathsf{E}\left[ \left( X_k-\tilde{X}_k \right) ^{\top}\left( X_k-\tilde{X}_k \right) \middle| Z^{k-1} \right]\nonumber \\ =&P\!\left(\! N_k=0 \middle| Z^{k-1} \!\right) \!f_k\left( Z^{k-1} \right)\! \left(\! f_k\left( Z^{k-1} \right) +P_{k|k-1}^{xx} \!\right) ^{-1}\!P_{k|k-1}^{xx}, \nonumber
				\end{align}
				where $P\left( N_k=0 \middle| Z^{k-1} \right)$ is given by \eqref{Eq.AppProbNk}. 
			
			As for the second term, \emph{i.e.}, the conditional mutual information in $l_k$, we first decompose it into three terms, as shown in \eqref{Eq.AppInf}. 
			\begin{figure*}
				\begin{align}\label{Eq.MUILoss1}
					\setcounter{equation}{38}
					&\int{\int{P\left( N_k=1 \middle| x_k,Z^{k-1} \right) p\left( x_k,y^k \middle| Z^{k-1} \right) \log p\left( y^k \middle| N_k=1,x_k,Z^{k-1} \right)}}dx_kdy^k \nonumber
					\\
					&=\int{p\left( x_k \middle| Z^{k-1} \right) P\left( N_k=1 \middle| x_k,Z^{k-1} \right) \int{p\left( y^k \middle| N_k=1,x_k,Z^{k-1} \right) \log p\left( y^k \middle| N_k=1,x_k,Z^{k-1} \right)}}dy^kdx_k \nonumber
					\\
					&\overset{(a)}{=}-\int{p\left( x_k \middle| Z^{k-1} \right) P\left( N_k=1 \middle| x_k,Z^{k-1} \right)}\left( \log \sqrt{\left| 2\pi P_{k|k}^{yy} \right|}+\frac{kn_y}{2} \right) dx_k \nonumber
					\\
					&\overset{(b)}{=}-\int{p\left( x_k \middle| Z^{k-1} \right) P\left( N_k=1 \middle| x_k,Z^{k-1} \right)}dx_k\left( \log \sqrt{\left| 2\pi \left( P_{k|k-1}^{yy}-P_{k|k-1}^{yx}\left( P_{k|k-1}^{xx} \right) ^{-1}P_{k|k-1}^{xy} \right) \right|}+\frac{kn_y}{2} \right) \nonumber
					\\
					&=-\int{p\left( x_k \middle| Z^{k-1} \right) \left( 1-P\left( N_k=0 \middle| x_k,Z^{k-1} \right) \right)}dx_k\left( \log \sqrt{\left| 2\pi \left( P_{k|k-1}^{yy}-P_{k|k-1}^{yx}\left( P_{k|k-1}^{xx} \right) ^{-1}P_{k|k-1}^{xy} \right) \right|}+\frac{kn_y}{2} \right) \nonumber
					\\
					&=-\left( 1-P\left( N_k=0 \middle| Z^{k-1} \right) \right) \left( \log \sqrt{\left| 2\pi \left( P_{k|k-1}^{yy}-P_{k|k-1}^{yx}\left( P_{k|k-1}^{xx} \right) ^{-1}P_{k|k-1}^{xy} \right) \right|}+\frac{kn_y}{2} \right) ,
				\end{align}
				\hrule
				\begin{align} \label{Eq.AppMUILoss2}
					&P\left( N_k=0 \middle| Z^{k-1} \right) \int{p\left( y^k \middle| N_k=0,Z^{k-1} \right) \log p\left( y^k \middle| N_k=0,Z^{k-1} \right)}dy^k \nonumber
					\\
					&=-P\left( N_k=0 \middle| Z^{k-1} \right) \left( \log \sqrt{\left| 2\pi P_{k|k}^{yy} \right|}+\frac{kn_y}{2} \right) \nonumber
					\\
					&=-P\left( N_k=0 \middle| Z^{k-1} \right) \left( \log \sqrt{\left| 2\pi \left(P_{k|k-1}^{yy}-P_{k|k-1}^{yx}\left( f_k\left( Z^{k-1} \right) +P_{k|k-1}^{xx} \right) ^{-1}P_{k|k-1}^{xy}\right) \right|}+\frac{kn_y}{2} \right).
				\end{align}
				\hrule
			\end{figure*}
			
			For the first term of the conditional mutual information in \eqref{Eq.AppInf}, we demonstrate that it is a function of $P_{k|k-1}$, as given in \eqref{Eq.MUILoss1}. Here, step $(a)$ follows from the entropy of the conditional Gaussian density $p\left( y^k \middle| N_k = 1, x_k, Z^{k-1} \right)$, step $(b)$ is derived using the update rules for the conditional variance $P_{k|k}^{yy}$, and $kn_y$ denotes the dimension of $Y^k$.
			\begin{figure*}
				\begin{align}\label{Eq.AppDisGrad1}
					\nabla _{\theta}\mathsf{E}\left[ \sum_{k=0}^K{l_D\left( X_k,\tilde{\pi}_{k,\varphi} \right)} \right] =\mathsf{E}\left[ \sum_{k=0}^K{\left( \mathrm{J}_{\tilde{\pi}_{k,\varphi}}\left( \theta \right)  \right) ^{\top}\nabla _{\tilde{\pi}_{k,\varphi}}l_D\left( X_k,\tilde{\pi}_{k,\varphi} \right)}+\left( \sum_{k=0}^K{l_D\left( X_k,\tilde{\pi}_{k,\varphi} \right)} \right) \left( \sum_{k=0}^K{\nabla _{\theta}\log \pi _{\theta}\left( N_k \middle| X^k,N^{k-1} \right)} \right) \right] .
				\end{align}
				\hrule
			\end{figure*}
			
			Similarly, we can compute the second term of the conditional mutual information in \eqref{Eq.AppInf} using \eqref{Eq.AppMUILoss2}.
			
			As for the last term of the conditional mutual information in \eqref{Eq.AppInf}, since $p\left( y^k \middle| Z^{k-1} \!\right)$ is a Gaussian density with the covariance $P_{k|k-1}^{yy}$, we have
			\begin{align}
				\int{\!p\!\left( y^k \middle| Z^{k-1} \!\right)\! \log p\!\left( y^k \middle| Z^{k-1} \!\right)}dy^k\!=-\!\log \! \sqrt{\left|\! 2\pi P_{k|k-1}^{yy} \!\right|}-\frac{kn_y}{2}. \nonumber
			\end{align}
			By reorganizing three terms, we finally obtain the one-step loss \eqref{Eq.OneStepLoss}. Since each term is conditional on $Z^{k-1}$, and determined by $f_k$ and $g_k$, we denote the one-step loss with $l_k\left(f_k, g_k, Z^{k-1} \right)$.
			\section{Proof of Lemma \ref{Lm.PG}}\label{App.Lm.PG}
			To update $\theta$, we fix $\phi$ and compute the gradient of the objective function  \eqref{Eq.ETParamOpt} for the sampler, which consists of two terms: the distortion and mutual information. 
			
			\textbf{1) The policy gradient of the distortion term:} We first consider the policy gradient of the distortion term as follows,
				\begin{align}
					&\nabla _{\theta}\mathsf{E}\left[ \sum_{k=0}^K{l_D\left( X_k,\tilde{\pi}_{k,\varphi} \right)} \right] \nonumber
					\\
					=&\mathsf{E}\!\left[\! \sum_{k=0}^K{\nabla _{\theta}l_D\!\left( X_k,\tilde{\pi}_{k,\varphi} \right)}+\left( \sum_{k=0}^K{\!l_D\left( X_k,\tilde{\pi}_{k,\varphi} \right)} \right)\! \nabla _{\theta}\log p_{\theta}\left( \tau \right) \!\right] \nonumber.
				\end{align}
				where $$\nabla _{\theta}l_D\left( X_k,\tilde{\pi}_{k,\varphi} \right) =\left( \mathrm{J}_{\tilde{\pi}_{k,\varphi}}\left( \theta \right) \right) ^{\top}\nabla _{\tilde{\pi}_{k,\varphi}}l_D\left( X_k,\tilde{\pi}_{k,\varphi} \right), $$ $\tau =\left\{ Y^K,X^K,N^K \right\} $, and $\mathrm{J}_{\tilde{\pi}_{k,\varphi}}\left(\theta\right)$ is the Jocabian matrix, which is an implicit function of $\theta$. To compute the policy gradient, we need to obtain the formulation of $\nabla _{\theta}\log p_{\theta}\left( \tau \right)$ and $\mathrm{J}_{\tilde{\pi}_{k,\varphi}}\left(\theta\right)$.
			
			\textbf{a) $\mathbf{\nabla _{\theta}\log p_{\theta}\left( \tau \right)}$: } With 
			\begin{align}
				p_{\theta}\!\left( \tau \right)\! =\!\prod_{k=0}^K{\!p\left(\! Y_k \middle| Y_{k-1} \!\right) p\!\left(\! X_k \middle| X_{k-1},Y_{k-1} \!\right) \pi _{\theta}\!\left(\! N_k \middle| X^k,N^{k-1} \!\right)},\nonumber
			\end{align}
			we have 
			\begin{align}
				\nabla _{\theta}\log p_{\theta}\left( \tau \right) =\sum_{k=0}^K{\nabla _{\theta}\log \pi _{\theta}\left( N_k \middle| X^k,N^{k-1} \right)},\nonumber
			\end{align}
			thus, the gradient turns to be \eqref{Eq.AppDisGrad1}.
			
			\textbf{b) $\mathbf{\mathrm{J}_{\tilde{\pi}_{k,\varphi}}\left( \theta \right) }$: }  
				To compute \( \mathrm{J}_{\tilde{\pi}_{k,\varphi}}\left( \theta \right) \), we first apply the chain rule, yielding  
				\begin{equation}  
					\mathrm{J}_{\tilde{\pi}_{k,\varphi}}\left( \theta \right) = \mathrm{J}_{\tilde{\pi}_{k,\varphi}}\left( \varphi \right) \mathrm{J}_{\varphi}\left( \theta \right),  
				\end{equation}  
				where \( \mathrm{J}_{\varphi}\left(\theta\right) \) is the Jacobian matrix. Since the mapping \( \varphi(\theta) \) is unknown, we cannot directly compute this Jacobian matrix. Inspired by \cite{fiez2020implicit}, we next derive \( \mathrm{J}_{\varphi}\left( \theta \right) \) using the implicit function theorem. 
			
			As \( \varphi \) represents the best response to \( \theta \), it satisfies
				\begin{align}
					\setcounter{equation}{41}
					\nabla _{\varphi}\tilde{L}\left( \theta ,\varphi \right) =\nabla _{\varphi}\mathsf{E}\left[ \sum_{k=0}^K{l_D\left( X_k,\tilde{\pi}_{k,\varphi} \right)} \right] =0.
				\end{align}
				Applying the implicit function theorem, we obtain  
				\begin{align}
					\nabla _{\theta}\nabla _{\varphi}\mathsf{E}\!\left[\! \sum_{k=0}^K{l_D\!\left(\! X_k,\tilde{\pi}_{k,\varphi} \!\right)} \!\right] \!+\!\nabla _{\varphi}^{2}\mathsf{E}\!\left[\! \sum_{k=0}^K{l_D\left( X_k,\tilde{\pi}_{k,\varphi} \right)} \!\right] \!\mathrm{J}_{\varphi}\left( \theta \right) \!=\!0. \nonumber
				\end{align}
				Thus, \( \mathrm{J}_{\varphi}\left( \theta \right) \) can be derived as shown in \eqref{Eq.AppJacMat1}.  
			\begin{figure*}
				\begin{align}\label{Eq.AppJacMat1}
					\setcounter{equation}{42}
					\mathrm{J}_{\varphi}\left( \theta \right) &=-\left( \nabla _{\varphi}^{2}\mathsf{E}\left[ \sum_{k=0}^K{l_D\left( X_k,\tilde{\pi}_{k,\varphi} \right)} \right] \right) ^{-1}\nabla _{\theta}\nabla _{\varphi}\mathsf{E}\left[ \sum_{k=0}^K{l_D\left( X_k,\tilde{\pi}_{k,\varphi} \right)} \right] \nonumber
					\\
					&=-\left( \nabla _{\varphi}^{2}\mathsf{E}\left[ \sum_{k=0}^K{l_D\left( X_k,\tilde{\pi}_{k,\varphi} \right)} \right] \right) ^{-1}\nabla _{\theta}\mathsf{E}\left[ \sum_{k=0}^K{\nabla _{\varphi}l_D\left( X_k,\tilde{\pi}_{k,\varphi} \right)} \right] \nonumber
					\\
					&=-\left( \nabla _{\varphi}^{2}\mathsf{E}\left[ \sum_{k=0}^K{l_D\left( X_k,\tilde{\pi}_{k,\varphi} \right)} \right] \right) ^{-1}\mathsf{E}\left[ \left( \sum_{k=0}^K{\nabla _{\varphi}l_D\left( X_k,\tilde{\pi}_{k,\varphi} \right)} \right) \otimes \left( \nabla _{\theta}\log p_{\theta}\left( \tau \right) \right)  \right] .
				\end{align}
				\hrule
			\end{figure*}
			
			\textbf{2) The gradient of the mutual information term:} 
			We finally study the gradient of the mutual information term as follows,
			\begin{align}
				\nabla _{\theta}I_{\theta}\left( Z^K;Y^K \right) =&\nabla _{\theta}\mathsf{E}\left[ \sum_{k=0}^K{\log \frac{p_{\theta}\left( Y^k \middle| Z^k \right)}{p_{\theta}\left( Y^k \middle| Z^{k-1} \right)}} \right] \nonumber
				\\
				=&\mathsf{E}\left[ \left( \sum_{k=0}^K{\log \frac{p_{\theta}\left( Y^k \middle| Z^k \right)}{p_{\theta}\left( Y^k \middle| Z^{k-1} \right)}} \right) \nabla _{\theta}\log p_{\theta}\left( \tau \right) \right] \nonumber\\&\;+\mathsf{E}\left[ \sum_{k=0}^K{\nabla _{\theta}\log \frac{p_{\theta}\left( Y^k \middle| Z^k \right)}{p_{\theta}\left( Y^k \middle| Z^{k-1} \right)}} \right] . \nonumber
			\end{align}
			where
			\begin{align}
				\mathsf{E}\left[ \nabla _{\theta}\log p_{\theta}\left( Y^k \middle| Z^k \right) \right] &=\mathsf{E}\left[ \frac{\nabla _{\theta}p_{\theta}\left( Y^k \middle| Z^k \right)}{p_{\theta}\left( Y^k \middle| Z^k \right)} \right] \nonumber
				\\
				&=\mathsf{E}\left[ \int{\nabla _{\theta}p_{\theta}\left( y^k \middle| Z^k \right)}dy^k \right] \nonumber
				\\
				&=\mathsf{E}\left[ \nabla _{\theta}\int{p_{\theta}\left( y^k \middle| Z^k \right)}dy^k \right] \nonumber
				\\
				&=0, \nonumber
			\end{align}
			and
			\begin{align}
				\mathsf{E}\left[ \nabla _{\theta}\log p_{\theta}\left( Y^k \middle| Z^{k-1} \right) \right] &=\mathsf{E}\left[ \frac{\nabla _{\theta}p_{\theta}\left( Y^k \middle| Z^{k-1} \right)}{p_{\theta}\left( Y^k \middle| Z^{k-1} \right)} \right] \nonumber
				\\
				&=\mathsf{E}\left[ \int{\nabla _{\theta}p_{\theta}\left( Y^k \middle| Z^{k-1} \right)}dy^k \right] \nonumber
				\\
				&=\mathsf{E}\left[ \nabla _{\theta}\int{p_{\theta}\left( Y^k \middle| Z^{k-1} \right)}dy^k \right] \nonumber
				\\
				&=0, \nonumber
			\end{align}
			Therefore, we have
			\begin{align}
				&\nabla _{\theta}I_{\theta}\left( Z^K;Y^K \right)\nonumber\\=&\mathsf{E}\!\left[\! \left(\! \sum_{k=0}^K{\!\log \frac{p_{\theta}\left( Y^k \middle| Z^k \right)}{p_{\theta}\left( Y^k \middle| Z^{k-1} \right)}} \!\right)\!\! \left(\! \sum_{k=0}^K{\!\nabla _{\theta}\log \pi _{\theta}\left( N_k \middle| X^k,N^{k-1} \right)} \!\right) \!\right] . \nonumber
			\end{align}
			
			\section{Proof of Lemma \ref{Lm.OptProblem} } \label{App.Lm.OptProblem}
			The mutual information can be expanded into multiple terms along time, \emph{i.e.},
				\begin{align}
					&I\left( Z^K;Y^K \right) \nonumber
					\\
					\overset{\left( a \right)}{=}&\sum_{k=0}^K{I\left( Z_k;Y^K \middle| Z^{k-1} \right)} \nonumber
					\\
					\overset{\left( b \right)}{=}&\sum_{k=0}^K{I\left( Z_k;Y^k \middle| Z^{k-1} \right) +I\left( Z_k;Y_{k+1}^{K} \middle| Z^{k-1},Y^k \right)}\nonumber
					\\
					\overset{\left( c \right)}{=}&\sum_{k=0}^K{I\left( Z_k;Y^k \middle| Z^{k-1} \right)} \nonumber
					\\ = & \sum_{k=0}^K\mathsf{E}\!\left[ \log \frac{p\left( Y^k \middle| Z^k \right)}{p\left( Y^k \middle| Z^{k-1} \right)} \right] \nonumber .
				\end{align}
				where $(a), (b)$ follow from the chain rule for mutual information, and $(c)$ is due to the fact that $Y^{K}_{k+1} \to \left(Z^{k-1},Y^k\right) \to Z_k$ is a Markov chain. Therefore, we have
				\begin{align}
					\!\!&\sum_{k=0}^{K}\mathsf{E} \left[\textit{l}_{D}\left(X_k,\tilde{\pi}^{\star}_k\left(Z^{k}\right)\right)\right]+\lambda I \left(Z^{K};Y^{K}\right)\nonumber\\ =&\sum_{k=0}^K{\mathsf{E}\left[ l_D\left( X_k,\tilde{\pi}_{k}^{\star}\left( Z^k \right) \right) \right] +\lambda \mathsf{E}\left[ \log \frac{p\left( Y^k \middle| Z^k \right)}{p\left( Y^k \middle| Z^{k-1} \right)} \right]}, \nonumber
					\\=&\mathsf{E}\!\!\left[\! \sum_{k=0}^K\!{\mathsf{E}\!\left[ l_D\!\left( X_k,\tilde{\pi}_{k}^{\star}\left( Z^k \right) \!\right) \!\middle| Z^{k-1} \right] \!\!+\!\lambda \mathsf{E}\!\!\left[ \! \log\! \frac{p\left( Y^k \middle| Z^k \right)}{p\!\left(\! Y^k \middle| Z^{k-1} \right)} \middle| Z^{k-1} \right]} \!\right] \!. \nonumber
			\end{align}
			\begin{figure*}
				\begin{align}\label{Eq.AppOneDistEq}
					\setcounter{equation}{43}
					\mathsf{E}\left[ l_D\left( X_k,\tilde{\pi}_{k}^{\star}\left( Z^k \right) \right) \middle| Z^{k-1} \right] 
					=&P\left( N_k=0 \middle| Z^{k-1} \right) \int{p\left( x_k \middle| N_k=0,Z^{k-1} \right) l_D\left( x_k,\tilde{\pi}_{k}^{\star}\left( Z_k=\!\emptyset,Z^{k-1} \right) \!\right)\! dx_k}\nonumber
					\\
					=&\int{\int{\int{a_k\left( N_k=0 \middle| x_k,m_k \!\right) b_k\left( x_k,m_k,y^k|Z^{k-1} \right) l_D\left( x_k,\tilde{\pi}_{k}^{\star}\left( Z_k=\emptyset,Z^{k-1} \right) \!\right)}dx_kdm_kdy^k}}\nonumber
					\\
					\overset{\bigtriangleup}{=}&\tilde{l}_D\left( \mathcal{A} _k,b_k,\tilde{\pi}_{k}^{\star}\left( Z_k=\emptyset,Z^{k-1} \right) \right) ,
				\end{align}
				\hrule
				\begin{align}\label{Eq.AppMIYKEq}
					&\mathsf{E}\left[ \log p\left( Y^k \middle| Z^k \right) \middle| Z^{k-1} \right] =\int{p\left( y^k \middle| N_k=0,Z^{k-1} \right) P\left( N_k=0 \middle| Z^{k-1} \right) \log p\left( y^k \middle| N_k=0,Z^{k-1} \right) dy^k} \nonumber
					\\
					&+\int{p\left( y^k \middle| N_k=1,x_k,Z^{k-1} \right) p\left( x_k \middle| N_k=1,Z^{k-1} \right) P\left( N_k=1 \middle| Z^{k-1} \right) \log p\left( y^k \middle| N_k=1,x_k,Z^{k-1} \right) dx_kdy^k},
				\end{align}
				\hrule
				\begin{align}\label{Eq.AppMIYKEq2}
					&\mathsf{E}\left[ \log p\left( Y^k \middle| Z^k \right) \middle| Z^{k-1} \right] \nonumber \\ &=\!\!\!\int{\!\!\!\left( \!\int{\!\!\!\int{\!\! a_k \! \left( N_k=0 \middle| x_k,m_k \right) p\!\left( x_k,m_k,y^k \middle| Z^{k-1} \right)}dx_kdm_k}\! \right)\! \log \! \left(\!\! \frac{\int\!\!{\int\!\!{a_k\left( N_k=0 \middle| x_k,m_k \right) p\left( x_k,m_k,y^k \middle| Z^{k-1} \right)}dx_kdm_k}}{\int\!\!{\int\!\!{\int\!\!{a_k\!\left( N_k=0 \middle| x_k,m_k \right)\! b_k\!\left( x_k,m_k,y^k|Z^{k-1}\! \right)}dx_kdm_kdy^k}}} \!\!\right) dy^k}\nonumber
					\\
					&+\!\!\int{\!\!\left(\! \int{\!\!a_k\!\left( N_k=1 \middle| x_k,m_k \right) p\!\left( x_k,m_k,y^k \middle| Z^{k-1} \right) dm_k}\!\! \right)\! \log\! \left( \frac{\int{a_k\left( N_k=1 \middle| x_k,m_k \right) p\left( x_k,m_k,y^k \middle| Z^{k-1} \right) dm_k}}{\int{\!\int{\!a_k\left( N_k=1 \middle| x_k,m_k \right) b_k\left( x_k,m_k,y^k|Z^{k-1} \right) dm_kdy^k}}} \right) dx_kdy^k},
				\end{align}
				\hrule
				\begin{align}\label{Eq.AppMIYKM1Eq}
					\mathsf{E}\left[ \log p\left( Y^k \middle| Z^{k-1} \right) \middle| Z^{k-1} \right] &=\int{p\left( y^k \middle| Z^{k-1} \right) \log p\left( y^k \middle| Z^{k-1} \right) dy^k} \nonumber
					\\
					&=\int{\!\left( \int{\!\!\int{\!b_k\left( x_k,m_k,y^k|Z^{k-1} \right)}dx_kdm_k} \!\right)\! \log \!\left(\! \int{\!\!\int{\!b_k\left( x_k,m_k,y^k|Z^{k-1} \right)}dx_kdm_k} \right) dy^k},
				\end{align}
				\hrule
				\begin{align}\label{Eq.AppMILossEq}
					&l_I\left( \mathcal{A} _k,b_k \right) =-\int{\!\left( \int{\!\!\int{\!b_k\left( x_k,m_k,y^k|Z^{k-1} \right)}dx_kdm_k}\! \right) \!\log \!\left( \!\int{\!\!\int{\!b_k\left( x_k,m_k,y^k|Z^{k-1} \right)}dx_kdm_k} \right) dy^k} \nonumber
					\\
					&+\!\!\!\int{\!\!\!\left( \!\int{\!\!\!\int{\!\!a_k\!\left( N_k=0 \middle| x_k,m_k \right) p\!\left( x_k,m_k,y^k \middle| Z^{k-1} \right)}dx_kdm_k}\! \right) \!\log \!\left( \!\!\frac{\int{\!}\!\int{\!}\!a_k\left( N_k=0 \middle| x_k,m_k \right) p\left( x_k,m_k,y^k \middle| Z^{k-1} \right) dx_kdm_k}{\int{\!}\!\int{\!}\!\int{\!}\!a_k\!\left( N_k=0 \middle| x_k,m_k \right) \!b_k\!\left( x_k,m_k,y^k|Z^{k-1}\! \right) dx_kdm_kdy^k}\!\! \right) dy^k} \nonumber
					\\
					&+\!\!\int{\!\!\left( \!\int{\!\!a_k\!\left( N_k=1 \middle| x_k,m_k \right) p\!\left( x_k,m_k,y^k \middle| Z^{k-1} \right) dm_k}\!\! \right) \!\log \!\left( \frac{\int{a_k\left( N_k=1 \middle| x_k,m_k \right) p\left( x_k,m_k,y^k \middle| Z^{k-1} \right) dm_k}}{\int{\!\int{\!a_k\left( N_k=1 \middle| x_k,m_k \right) b_k\left( x_k,m_k,y^k|Z^{k-1} \right) dm_kdy^k}}} \right) dx_kdy^k}. 
				\end{align}
				\hrule
			\end{figure*}
			\section{Proof of Lemma \ref{Lm.beliefLoss}}\label{App.Lm.beliefLoss}
			Let $b_k$ be the belief state, \emph{i.e.}, $b_k\left( x_k,m_k,y^{k} \right) =p\left( x_k,m_k,y^{k}|Z^{k-1} \right) $. Given the policy collection $\mathcal{A}_k$ and the sent information $Z_k$, if $N_k$ is zero, then we have $M_{k+1}=\{M_k,X_k\}$, and
			$b_{k+1}$ can be updated as follows
			\begin{align}
				&b_{k+1}\left( x_{k+1},m_{k+1},y^{k+1} \right) \nonumber\\ =& p\left( x_{k+1},m_k,x_k,y^{k+1}|N_k=0,Z^{k-1} \right) \nonumber
				\\
				=&\frac{p\!\left(\! x_{k+1} \middle| x_k,y_k \!\right) p\!\left(\! y_{k+1} \middle| y_k \!\right) b_k\!\left(\! x_k,m_k,y^k|Z^{k-1} \!\right) a_k\!\left(\! N_k\!=\!0 \middle| x_k,m_k \!\right)}{P\left( N_k=0 \middle| Z^{k-1} \right)}\nonumber
				\\
				=&\frac{p\!\left(\! x_{k+1} \middle| x_k,y_k \!\right) p\!\left(\! y_{k+1} \middle| y_k \!\right) b_k\!\left(\! x_k,m_k,y^k|Z^{k-1} \!\right) a_k\!\left(\! N_k\!=\!0 \middle| x_k,m_k \!\right)}{\int{\int{\int{a_k\left( N_k=0 \middle| x_k,m_k \right) b_k\left( x_k,m_k,y^k|Z^{k-1} \right)}dx_kdm_kdy^k}}}. \nonumber
			\end{align}  
			If $N_k$ is one, then we have $M_{k+1}=M_k$, 
			and $b_{k+1}$ can be updated via
			\small
			\begin{align}
				&b_{k+1}\left( x_{k+1},m_{k+1},y^{k+1} \right)\nonumber\\ =&p\left( x_{k+1},m_k,y^{k+1}|N_k=1,X_k,Z^{k-1} \right) \nonumber
				\\
				=&\frac{p\!\left(\! x_{k+1} \middle| X_k,y_k \!\right) p\!\left(\! y_{k+1} \middle| y_k \!\right) b_k\!\left(\! X_k,m_k,y^k|Z^{k-1} \!\right) a_k\!\left(\! N_k=1 \middle| X_k,m_k \!\right)}{P\left( N_k=1 \middle| X_k,Z^{k-1} \right) p\left( X_k \middle| Z^{k-1} \right)}\nonumber
				\\
				=&\frac{p\!\left(\! x_{k+1} \middle| X_k,y_k \!\right) p\!\left(\! y_{k+1} \middle| y_k \!\right) b_k\!\left(\! X_k,m_k,y^k|Z^{k-1} \!\right) a_k\!\left(\! N_k=1 \middle| X_k,m_k \!\right)}{\int{\int{a_k\left( N_k=1 \middle| X_k,m_k \right) b_k\left( X_k,m_k,y^k|Z^{k-1} \right) dm_kdy^k}}}.\nonumber
			\end{align}  
			\normalsize
			
			In \eqref{Eq.AppOneDistEq}, we show the conditional reconstruction error is a function of $\mathcal{A}_k$, $b_k$.
			
			Next, we analyze the conditional mutual information. $\mathsf{E}\left[ \log p\left( Y^k \middle| Z^k \right) \middle| Z^{k-1} \right]$ can be written as \eqref{Eq.AppMIYKEq}, where
			\small
			\begin{align}
				&p\left( y^k \middle| N_k=0,Z^{k-1} \right) \nonumber
				\\
				=&\frac{\int{\!\int{\!\pi_k\left(\! N_k=0 \middle| x_k,m_k,Z^{k-1} \right) p\left( x_k,m_k,y^k \middle| Z^{k-1} \right)}dx_kdm_k}}{P\left( N_k=0 \middle| Z^{k-1} \right)} \nonumber
				\\
				=&\frac{\int{\!\int{\!a_k\left(\! N_k=0 \middle| x_k,m_k \right) p\left( x_k,m_k,y^k \middle| Z^{k-1} \right)}dx_kdm_k}}{P\left( N_k=0 \middle| Z^{k-1} \right)} \nonumber
				\\
				=&\frac{\int{\int{a_k\left( N_k=0 \middle| x_k,m_k \right) p\left( x_k,m_k,y^k \middle| Z^{k-1} \right)}dx_kdm_k}}{\int{\int{\int{a_k\left( N_k=0 \middle| x_k,m_k \right) b_k\left( x_k,m_k,y^k|Z^{k-1} \right)}dx_kdm_kdy^k}}}, \nonumber
			\end{align}
			\normalsize
			and
			\begin{align}
				&p\left( y^k \middle| N_k=1,x_k,Z^{k-1} \right) \nonumber\\=&\frac{\int{a_k\left( N_k=1 \middle| x_k,m_k \right) p\left( x_k,m_k,y^k \middle| Z^{k-1} \right) dm_k}}{p\left( x_k \middle| N_k=1,Z^{k-1} \right) P\left( N_k=1 \middle| Z^{k-1} \right)} \nonumber
				\\
				=&\frac{\!\int{a_k\left( N_k=1 \middle| x_k,m_k \right) p\left( x_k,m_k,y^k \middle| Z^{k-1} \right) dm_k}}{\int{\!\!\int{\!\!a_k\!\left( N_k=1 \middle| x_k,m_k \right) \!b_k\!\left( x_k,m_k,y^k|Z^{k-1} \right) dm_kdy^k}}}. \nonumber
			\end{align}
			Thus, we have \eqref{Eq.AppMIYKEq2}, which depends on $b_k$ and $\mathcal{A}_k$. Similarly, $\mathsf{E}\left[ \log p\left( Y^k \middle| Z^{k-1} \right) \middle| Z^{k-1} \right]$ is a function of $b_k$ as shown in \eqref{Eq.AppMIYKM1Eq}.
			
			By reorganizing these terms, we show $l_I\left( \mathcal{A} _k,b_k \right)$ depends on $\mathcal{A}_k$ and $b_k$ in \eqref{Eq.AppMILossEq}.
			
		\end{document}